\documentclass[12pt,reqno]{amsart}
\usepackage{fullpage}
\usepackage{times}
\usepackage{graphicx}
\usepackage{amsmath,amsthm,amsfonts,amscd,amssymb}
\usepackage{pst-node}
\usepackage{pst-plot}
\usepackage{fullpage}
\usepackage[all]{xy}

\usepackage{tikz, braids}
\usetikzlibrary{decorations.pathreplacing, arrows}

\newtheorem{thm}{Theorem}[section]
\newtheorem{cor}[thm]{Corollary}
\newtheorem{defn}[thm]{Definition}
\newtheorem{lem}[thm]{Lemma}
\newtheorem{rem}[thm]{Remark}
\newtheorem{prop}[thm]{Proposition}
\newtheorem{ex}[thm]{Example}
{ \theoremstyle{remark} }

\numberwithin{thm}{section}
\numberwithin{equation}{section}

\newcommand{\myop}[1]{\operatorname{#1}}

\newcommand{\Irr}{\myop{Irr}}
\newcommand{\Tr}{\myop{Tr}}

\newcommand{\Hom}{\myop{Hom}}

\newcommand{\N}{\mathbb N}
\newcommand{\C}{\mathbb C}

\newcommand{\G}{\mathbb G}

\newcommand{\TL}{\text{TL}}
\newcommand{\mc}{\mathcal}

\newcommand{\eps}{\varepsilon}

\newcommand{\la}{\langle}
\newcommand{\ra}{\rangle}

\newcommand{\Comp}{\mathbb{C}}

\newcommand{\g}{\mathbb{G}}

\newcommand{\n}{\mathbb{N}}

\global\long\def\tp{\mathop{\xymatrix{*+<.7ex>[o][F-]{\scriptstyle \top}}
 } }

\begin{document}

\title[Temperley-Lieb channels]{Temperley-Lieb quantum channels}

\author {Michael Brannan}
\address{Michael Brannan,
Department of Mathematics,
Mailstop 3368, Texas A\&M University, 
College Station, TX 77843-3368, USA}
\email{mbrannan@math.tamu.edu}

\author {Beno\^\i{}t Collins}
\address{Beno\^\i{}t Collins,
Department of Mathematics, Kyoto University,
Kyoto 606-8502, Japan }
\email{collins@math.kyoto-u.ac.jp}

\author{Hun Hee Lee}
\address{Hun Hee Lee, Department of Mathematical Sciences and the Research Institute of Mathematics, Seoul National University, Gwanak-ro 1, Gwanak-gu, Seoul 08826, Republic of Korea}
\email{hunheelee@snu.ac.kr}

\author{Sang-Gyun Youn}
\address{Sang-Gyun Youn, 
Department of Mathematical Sciences, Seoul National University, 
GwanAkRo 1, Gwanak-Gu, Seoul 08826, South Korea}
\email{yun87654@snu.ac.kr }

\maketitle

\begin{abstract}
We study a class of quantum channels arising from the representation theory of compact quantum groups that we call Temperley-Lieb quantum channels. These channels simultaneously extend those introduced in \cite{BrCo18b}, \cite{Al14}, and \cite{LiSo14}.
(Quantum) Symmetries in quantum information theory arise naturally from many points of view, providing an important source of new examples of quantum phenomena, and also serve as useful tools to simplify or solve important problems. This work provides new applications of quantum symmetries
in quantum information theory.  Among others, we study entropies and capacitites of Temperley-Lieb channels, their (anti-) degradability, PPT and entanglement breaking properties, as well as the behaviour of their tensor products with respect to entangled inpurs. Finally we compare the Tempereley-Lieb channels with the (modified) TRO-channels recently introduced in \cite{GaJuLa16}.
\end{abstract}

\section{Introduction}\label{sec-intro}

A fundamental problem in (quantum) information theory is to understand the capacity of a noisy communications channel.  In the quantum world, this is harder, because there are many notions of capacities, non-trivial additivity questions related to these capacities, and a very poor understanding of the behaviour of quantum channels under the operation of tensoring. The non-trivial channels for which 
many entropic or capacity related quantities can be computed and be of non-trivial value or interest are rather scarce.
One reason for this paucity is that many quantities are defined with minimizers, and many properties (e.g. PPT, entanglement breaking property (shortly, EBT), degradability and so on) rely on the existence of auxiliary objects or computations of tensors that are close to impossible to 
describe effectively without additional conceptual assumptions on the quantum channel.
 
One of the most natural  (and to our mind,  underrated) property of a quantum channel  is to have some sort of group symmetry. 
In this paper, we will focus on quantum channels which feature symmetries with respect to structures which are more general than groups: compact quantum groups. 
For example, the notion of a covariant quantum channel channel with respect to a compact group action was introduced in many contexts (\cite{WeHo02, DaFuHo06, MoStDa17, Al14, LiSo14, Rit05}) but these properties have not been extensively used from the analysis point of view of quantum information theory (shortly, QIT) such as estimating quantities. In addition, most of the time, the covariance under consideration is with respect to the most elementary group representations, e.g., the basic representation of a matrix group $G \subset M_n(\C)$ on $\C^n$.
The principal reason behind the restriction to the basic representations so far is  that the symmetries involved and the analysis behind many aspects of representation theory are not well-understood to the degree required to estimate important quantities.  Nonetheless, it was observed in many places that such symmetries can be useful (e.g. \cite{MuReWo16,HaMu15,Sc05, DaFuHo06, KoWe09, SaWoPeCi09, MoStDa17}, etc). See also \cite{CoOsSa18} for a covariant characterization of $k$-positive maps. 

The first systematic attempt to remedy this limitation was conducted by Al Nuwairan \cite{Al14} in the context of $SU(2)$ symmetries. Here, Al Nuwairan investigated quantum channels arising from the intertwining isometries of the irreducible decomposition of the tensor product of two irreducible representations of $SU(2)$, which we will call $SU(2)$-Temperley-Lieb quantum channels (shortly, $SU(2)$-TL-channels). Thanks to the well-known $SU(2)$-Clebsch-Gordan formulas, explicit results could be obtained and it turned out that $SU(2)$-TL-channels play important roles of describing general $SU(2)$-covariant quantum channels. However, from the perspective of entanglement theory, the performance of $SU(2)$-TL-channels was not spectacular.
Subsequently, \cite{BrCo18b} considered a quantum extension of $SU(2)$-TL-channels using irreducible representations of free orthogonal quantum groups, which we call $O^+_N$-TL-channels in this paper, and noticed that a notion of rapid decay was exactly the concept needed to estimate precisely the entanglement in a highly entangled setup. 
The main idea was to replace group symmetries by quantum group symmetries especially for the free orthogonal quantum group $O_N^+$ case, whose main advantage is that it allows to remain in a well-understood C$^\ast$-tensor category (the Temperley-Lieb category) which facilitates very explicit computations and estimates. 

The present work undertakes a much more systematic study of $SU(2)$-TL-channels and $O^+_N$-TL-channels, and compares their various information theoretic properties.
One important achievement of this paper is that the minimum output entropy (shortly, MOE) $H_{\min}$, the one-shot quantum capacity $Q^{(1)}$ and the Holevo capacity $\chi$ can be estimated, and that these estimates are asymptotically sharp as $N$ becomes big, in the case of $O_N^+$-TL-channels. More generally, the main results of this paper are summarized below in the following table:

\begin{table}[h!]
  \begin{center}
    \caption{Summary of results.}
    \label{tab:table1}
    \begin{tabular}{l|c|r} 
       \textbf{Properties}$\backslash$\textbf{Channels} & $O^+_N$-TL-ch. [sec. \ref{sec:moe-cap}, sec. \ref{sec:EBP-PPT}] & $SU(2)$-TL-ch.  [sec. \ref{sec:EBP-PPT}]\\
        \hline
     $H_{\min}$ & asympt. sharp & \cite{Al13}\\
      \hline
      $Q^{(1)}$ and $\chi$ & asympt. sharp & rough estimates\\
       \hline
      EBT & No except for the lowest weight & complete \\
       \hline
      PPT & No except for the lowest weight with $N \gg 1$ & complete \\
       \hline
     (Anti-)Degradability & No except for the lowest weight with $N \gg 1$ &  partial results \\
      \hline
      $C$ (classical capacity) & $C\le (2+\eps)\chi$ with $N \gg 1$ & ? (open)\\
       \hline
      Equivalence to TRO ch.& ? (open)  & No in general [sec. \ref{sec:TRO}]\\

    \end{tabular}
  \end{center}
\end{table}

The term TRO in the above will be clarified later in the introduction and in section \ref{sec:TRO} with more details.
As it appears from the above table, many interesting and unexpected phenomena are unveiled, which we find counterintuitive, 
and whose proof boils down to an extensive case analysis. 
Just to mention a few: 
\begin{itemize}
\item 
Many non-trivial results can be obtained about the degradability and anti-degradability of the covariant quantum channels. 
To the best of our knowledge, although these notions are really important to estimate capacities (and we use such results), there are almost no non-trivial examples in the literature of quantum channels for which one can assess the  degradability and anti-degradability. Our computation is possible thanks to averaging methods 
stemming from (quantum) group invariance. 
\item 
In most cases, $O_N^+$-TL-channels with large $N$ have a highly non-trivial structure. Indeed, they are not PPT, not degradable, not anti-degradable except for the possibility of lowest weight subrepresentations, which we still have not settled. 
Moreover, we present a complete list for EBT and PPT for $SU(2)$-TL-channels and it turns out that the notions of PPT and EBT are actually equivalent in the case of $SU(2)$. One important ingredient here is the diagrammatic calculus for Temperley-Lieb category covered in Section \ref{subsec:Diagram}.
\item
On the other hand, we reveal unexpected results on (anti-)degradability of $SU(2)$-TL-channels. We show that they are degradable for extremal cases such as lowest or highest weight, whereas it is not 
true for other intermediate cases. Indeed, we provide an example of a non-degradable $SU(2)$-TL-channel in low dimensions (see Example \ref{ex:non-deg-non-antideg}). 
\end{itemize}

One crucial point in QIT is that it is often unavoidable to consider tensor products of quantum channels, and in general, computations in tensor products become very involved. However when the channels have nice symmetries, as we show in this paper, 
computations can remain tractable, even in non-trivial cases. The main techinical tool is an application of diagrammatic calculus explained in Section \ref{subsec:Diagram}, which can be applied to $O^+_N$-TL-channels, see Section \ref{sec:tensor} for the details.

Finally, TL-channels bear some resemblance with another important family of operators introduced by \cite{GaJuLa16}, called TRO-channels and their modified versions. Here, TRO refers to ternary ring of operators and name ``TRO-channel'' comes from the fact that its Stinespring space, i.e. the range of the Stinespring isometry actually has a TRO structure. Examples of TRO-channels include random unitary channels from regular representations of finite (quantum) groups and generalized dephasing channels \cite{GaJuLa16}. While the authors were preparing this manuscript and discussing it for the first time publicly, the question of how our TL-channels compare to TRO channels was posed (and, in particular, whether or not TL implies TRO).   The answer is that these classes of channels bear important differences, as explained in section \ref{sec:TRO}.

This paper is organized as follows. 
After this introduction, section \ref{sec:preliminaries} provides some background and reminders about quantum channels and compact quantum groups.
Section \ref{sec:TL} recalls some details on free orthogonal quantum groups and their associated representation theory. Then, we introduce Tempereley-Lieb quantum channels (shortly, TL-channels) and collect some details on their associated diagrammatic calculus. 
Section \ref{sec:moe-cap} contains results about the entropies and capacities of TL-channels. 
Then, section \ref{sec:EBP-PPT} addresses the property of entanglement breaking and PPT for TL-channels. 
Section \ref{sec:tensor} shows that $O^+_N$-TL-channels (unlike most `structureless' quantum channels) behave very well under tensor products. Finally, section \ref{sec:TRO} addresses the question of comparing TL-channels with Junge's (modified) TRO-channels.

\subsection*{Acknowledgements} MB's research was supported by NSF grant DMS-1700267.  BC's research was supported by JSPS KAKENHI 17K18734, 17H04823, 15KK0162. HHL and SY's research was supported by the Basic Science Research Program through the National Research Foundation of Korea (NRF) Grant NRF-2017R1E1A1A03070510 and the National Research Foundation of Korea (NRF) Grant funded by the Korean Government (MSIT) (Grant No.2017R1A5A1015626).

The authors are grateful to Marius Junge for useful comments and discussons during various stages of preparation of this manuscript.

\section{Preliminaries}\label{sec:preliminaries}

\subsection{Quantum channels and their information theoretic quantities}
 
Here, we are only interested in quantum channels based on finite dimensional Hilbert spaces. Recall that a quantum channel is a linear completely positive trace-preserving (shortly, CPTP) map $\Phi: B(H_A) \to B(H_B)$. It is well-known that there is a so called {\it Stinespring isometry} $V : H_A \to H_B \otimes H_E$ such that
	$$\Phi(\rho) = (\iota \otimes {\rm Tr}_E)(V \rho V^*),\; \rho \in B(H_A),$$
where ${\rm Tr}_E$ refers to the trace on $B(H_E)$. For a given Stinespring isometry $V$ we can consider the complementary channel $\tilde{\Phi}:B(H_A) \to B(H_E)$ of $\Phi$ given by
	$$\tilde{\Phi}(\rho) = ({\rm Tr}_B \otimes \iota)(V \rho V^*),\; \rho \in B(H_A).$$

For each quantum channel there are several important information theoretic quantities, which we recall in the following.

\begin{defn}
Let $\Phi:B(H_A)\rightarrow B(H_B)$ be a quantum channel.

\begin{enumerate}

\item The Holevo capacity $\chi(\Phi)$ is defined by
\[\chi (\Phi) := \max \Big \{ H(\Phi \big (\sum_x p_x \rho_x\big ))-\sum_x p_x H(\Phi(\rho_x)) \Big \},\]
where the maximum runs over all possible choice of ensemble of quantum states $\{(p_x), (\rho_x)\}$ on $H_A$ and $H(\cdot)$ refers to the von Neumann entropy of a state $\rho \in B(H_A)$.

\item The ``one-shot'' quantum capacity $Q^{(1)}(\Phi)$ is defined by
\[Q^{(1)}(\Phi) := \max \{ H(\Phi(\rho)) - H(\tilde{\Phi}(\rho)) \}\]
where the maximum runs over all quantum states $\rho$ in $B(H_A)$. Note that the definition is independent of the choice of Stinespring isometry which determines the complementary channel $\tilde{\Phi}$.

\item The classical capacity $C(\Phi)$ and the quantum capacity $Q(\Phi)$ are obtained by the regularizations of the Holevo capacity and the ``one-shot'' quantum capacity, respectively, as follows.
\[C(\Phi)=\lim_{n \to \infty} \frac{\chi (\Phi^{\otimes n})}{n},\;\; Q(\Phi) =\lim_{n\to \infty} \frac{Q^{(1)}(\Phi^{\otimes n})}{n}.\]

\item The minimum output entropy (MOE) $H_{\min}(\Phi)$ given by
	$$H_{\min}(\Phi):=\min_{\rho} H(\Phi(\rho)),$$
where the minimum runs over all quantum states $\rho$ in $B(H_A)$.

\end{enumerate}
\end{defn}

\begin{rem}
The two quantities $\chi$ and $H_{\min}$ are closely related. In general, we have the following for a quantum channel $\Phi:B(H_A)\rightarrow B(H_B)$.
	\begin{equation}\label{eq-Holevo-MOE}
		\chi(\Phi) \le \log d_B - H_{\min}(\Phi),
	\end{equation}
where $d_B$ refers to the dimension of $H_B$ \cite{Hol-book}.	
\end{rem}

The regularization precedure for the classical capacity and the quantum capacity causes serious difficulties for the calculations of capacities in general. There are, however, some properties of channels that allow us to simplify the calculation, which we present below.

\begin{defn} Let $\Phi:B(H_A)\rightarrow B(H_B)$ be a quantum channel with the complimentary channel $\tilde{\Phi}:B(H_A) \to B(H_E)$.
\begin{enumerate}
\item We say that $\Phi$ is degradable (resp. anti-degradable) if there exists a channel $\Psi:B(H_B)\rightarrow B(H_E)$ (resp. $\Psi:B(H_E)\rightarrow B(H_B))$ such that $\widetilde{\Phi}=\Psi\circ \Phi$ (resp. $\Phi= \Psi\circ \widetilde{\Phi})$.
\item We say that $\Phi$ is entanglement-breaking (shortly, EBT) if there exist a probability distribution $(p_x)_x$ and product states $\rho_x^B\otimes \rho_x^A \in B(H_B \otimes H_A)$ such that the Choi matrix of $\Phi$, $\displaystyle C_{\Phi} :=\frac{1}{d_A}\sum_{i,j=1}^{d_A}\Phi(e_{ij})\otimes e_{ij}$ is given by $\displaystyle C_{\Phi} = \sum_x p_x \rho_x^B\otimes \rho_x^A.$

\item We say that $\Phi$ is PPT (positive partial transpose) if $(T_B\otimes \iota) C_\Phi$ is a positive matrix in $B(H_B \otimes H_A)$, equivalently if $T_B\circ \Phi$ is also a channel where $T_B$ is the transpose map on $B(H_B)$.

\item We say that $\Phi$ is {\it bistochastic} if $\Phi(\frac{1_A}{d_A}) = \frac{1_B}{d_B}$.

\end{enumerate}
\end{defn}

From the definition it is clear that EBT channels are PPT and by \cite[Corollary 10.28]{Hol-book} they are also  anti-degradable. Note that we have the following consequences of the above properties.

\begin{prop}\label{prop:implications}
Let $\Phi:B(H_A)\rightarrow B(H_B)$ be a quantum channel.
	\begin{enumerate}
		\item \cite{DeSh05} If $\Phi$ is degradable, then $Q(\Phi) = Q^{(1)}(\Phi)$.
		\item \cite{HoHoHo96,Pe96,Hol-book} If $\Phi$ is PPT or anti-degradable, then $Q(\Phi) = Q^{(1)}(\Phi) = 0$.
		\item \cite{Sh02} If $\Phi$ is EBT, then $C(\Phi) = \chi(\Phi)$.
	\end{enumerate}
\end{prop}

Some bistochastic channels have the following straightforward capacity estimates.
\begin{prop}\label{prop-bistochastic-estimates}
Let $\Phi: B(H_A) \to B(H_B)$ be a bistochastic quantum channel with a Stinespring isometry $V: H_A \to H_B \otimes H_E$. Suppose further that its complementary channel $\tilde{\Phi}$ is also bistochastic, then we have
	\begin{equation}
		\log \frac{d_B}{d_E} \le Q^{(1)}(\Phi) \le C(\Phi) \le \min \{ \log d_A, \log d_B, \log \frac{d_A d_B}{d_E} \}.
	\end{equation}
\end{prop}
\begin{proof}
We first observe that positivity of $\Phi$ tells us
	$$||\Phi||_{S^1(H_A) \to B(H_B)} \le ||\Phi||_{B(H_A) \to B(H_B)} = ||\Phi(1_A)||_{B(H_B)} = \frac{d_A}{d_B}.$$
Since $\Phi^{\otimes n}$ is also bistochastic, we also have $||\Phi^{\otimes n}||_{S^1(H_A^{\otimes n}) \to B(H_B^{\otimes n})} \le \big(\frac{d_A}{d_B}\big)^n$. Thus, we have
	$$H_{\min}(\Phi^{\otimes n}) = \min_{\rho} H(\Phi^{\otimes n}(\rho)) \ge -\log ||\Phi^{\otimes n}||_{S^1(H_A^{\otimes n}) \to B(H_B^{\otimes n})} \ge n \log \frac{d_B}{d_A}.$$
Note also that $H_{\min}(\Phi^{\otimes n}) = H_{\min}(\widetilde{\Phi^{\otimes n}}) = H_{\min}(\tilde{\Phi}^{\otimes n}) \ge n \log \frac{d_E}{d_A}$ 	
so that we have
	\begin{align*}
	\chi(\Phi^{\otimes n})
	& \le  \log d_B^n - H_{\min}(\Phi^{\otimes n})\\
	& \le n \log d_B - n \cdot \max \{ \log \frac{d_B}{d_A},  \log \frac{d_E}{d_A}\}\\
	& = n \cdot \min\{\log d_A, \log \frac{d_Ad_B}{d_E}\}.
	\end{align*}
Thus, we have
	$$C(\Phi) = \lim_{n\to \infty}\frac{\chi(\Phi^{\otimes n})}{n} \le \min \{ \log d_A, \log d_B, \log \frac{d_A d_B}{d_E} \}$$
together with the obvious estimate $\chi(\Phi^{\otimes n}) \le n\cdot \log d_B$.

The lower bound is direct from the definition of the ``one-shot'' quantum capacity.
	$$Q^{(1)}(\Phi) \ge H(\Phi(\frac{1_A}{d_A})) - H(\tilde{\Phi}(\frac{1_A}{d_A})) \ge H(\frac{1_B}{d_B}) - \log d_E = \log \frac{d_B}{d_E}.$$
\end{proof}

\subsection{Compact quantum groups and their representations}

A {\it compact quantum group} is a pair  $\g=(C(\g),\Delta)$ where $C(\g)$ is a unital $C^*$-algebra and $\Delta:C(\g)\rightarrow C(\g)\otimes_{\min}C(\g)$ is a unital $*$-homomorphism satisfying that (1) $(\Delta\otimes \iota)\Delta= (\iota\otimes \Delta)\Delta$ and (2) each of the spaces $\mathrm{span}\left \{\Delta(a)(1\otimes b):a,b\in C(\g) \right\}$ and $\mathrm{span}\left \{\Delta(a)(b\otimes 1):a,b\in C(\g) \right\}$ are dense in $C(\g)\otimes_{\min}C(\g)$. It is well known that every compact quantum group has the (unique) {\it Haar state} $h$, which is a state on $C(\g)$ such that $(\iota \otimes h)\Delta= h(\cdot )1=(h\otimes \iota)\Delta.$ If the Haar state $h$ is tracial, i.e. $h(ab)=h(ba)$ for all $a,b\in C(\g)$, then $\g$ is said to be of {\it Kac type}. 

A (finite dimensional) {\it representation} of $\g$ is a pair $(u,H_u)$ where $H_u$ is a finite dimensional Hilbert space and $u=(u_{i,j})_{1\leq i,j\leq d_u}\in B(H_u)\otimes C(\g)$ such that
$\displaystyle \Delta(u_{i,j})=\sum_{k=1}^{d_u}u_{i,k}\otimes u_{k,j}$ for all $1\leq i,j\leq d_u$. Here, $d_u$ refers to the dimension of $u$. The representation $u$ is called {\it unitary} if it further satisfies $u^*u=1_u\otimes 1=uu^*$. Whenever we have a unitary representation $(u,H_u)$ of $\G$ we obtain a so-called {\it $\G$-action} on $B(H_u)$
	\begin{equation}\label{eq-G-action}
	\beta_u: B(H_u) \to B(H_u) \otimes C(\G), \quad x \mapsto u(x\otimes 1 )u^*.
	\end{equation}
For given unitary representations $v$ and $w$, we say that a linear map $T:B(H_v)\rightarrow B(H_w)$ {\it intertwines} $v$ and $w$ if 
$$(T\otimes 1) v = w (T\otimes 1)$$
and denote by $\mathrm{Hom}_\G(v,w)$ (simply, $\mathrm{Hom}(v,w)$) the space of {\it intertwiners}. If $\mathrm{Hom}(v,w)$ contains an invertible intertwiner, then $v$ and $w$ are said to be {\it equivalent}. A unitary representation $(v,H_v)$ is called {\it irreducible} if $\mathrm{Hom}(v)=\mathrm{Hom}(v,v)=\Comp\cdot 1_v$ and we denote by $\mathrm{Irr}(\g)$ the set of all irreducible unitary representations of $\g$ up to equivalence. 

When we fix a representative $u^{\alpha}= [u^{\alpha}_{ij}]_{i,j=1}^{d_\alpha} \in M_{d_{\alpha}} (C(\G))$ for each $\alpha \in \mathrm{Irr}(\g)$, the Peter-Weyl theory for compact quantum groups says the space $\mathrm{Pol}(\g) :=  \mathrm{span}\{u^{\alpha}_{ij}: \alpha \in \mathrm{Irr}(\g), 1\le i,j\le d_\alpha \}$ is a subalgebra of $C(\G)$ containing all the information on the quantum group $\G$. In particular, it hosts the map $S$ called the {\it antipode} determined by the formula
\[ S(u^{\alpha}_{ij})= (u_{ji}^{\alpha})^*, \;\;  \alpha \in \mathrm{Irr}(\g), 1\le i,j\le d_\alpha.\]

For representations $v=(v_{ij})$ and $w=(w_{kl})$ we define its {\it tensor product} $v \tp w$ by
$$ v\tp w = \sum_{i,j=1}^{d_v}\sum_{k,l=1}^{d_w} e_{ij}\otimes e_{kl}\otimes v_{ij}w_{kl} \in B(H_v)\otimes B(H_w)\otimes C(\g) .$$ 
Then {\it the representation category} consisting of unitary representations as objects and intertwiners as morphisms is a {\it strict $C^*$-tensor category} under the natural adjoint operation $Hom(v,w)\rightarrow Hom(w,v), T\mapsto T^*$, and the tensor product $\tp$. 
It is well known that any finite dimensional representation decomposes into a direct sum of irreducible representations, so that we have
	$$v \tp w \cong \oplus^N_{i=1}u_i.$$
In case $u$ is a component of the irreducible decomposition of $v \tp w$ we write $u\subset v \tp w$.

For a given unitary representation $(v,H_v)$ we consider the map $j:B(H)\rightarrow B(\overline{H})$ defined by $j(T)\overline{\xi}=\overline{T^*\xi}$. Then the {\it contragredient representation} of $v$ is given by 
$$ v^c=(v_{ij}^*)_{1\leq i,j\leq d_v}=(j\otimes \iota)(v^{-1})\in B(\overline{H})\otimes C(\g).$$
The contragredient representation $v^c$ is unitary if $\g$ is of Kac type.

For each compact quantum group $\G$ we have its opposite version $\G^{\rm op}$ with the same algebra $C(\G^{\rm op}) = C(\G)$, but with the flipped co-multiplication $\Delta_{\rm op} = \Sigma \circ \Delta$, where $\Sigma$ is the flip map on $C(\G) \otimes_{\min} C(\G)$. Then, for any unitrary representation $u= (u_{ij}) \in B(H_u)\otimes C(\G)$ of $\G$ we have an associated representation $u^* = (u^*_{ji}) \in B(H_u)\otimes C(\G)$ of $\G^{\rm op}$.

\subsection{Clebsch-Gordan channels}
Let $\G$ be a compact quantum group and $(u,H_u)$, $(v,H_v)$ and $(w,H_w)$ be unitary irreducible representations of $\G$ such that $u \subset v \tp w$, which gives us its intertwining isometry $\alpha_u^{v,w}:H_u \to H_v \otimes H_w$. By using $\alpha_u^{v,w}$ as the Stinespring isometry we get the following complementary pair of quantum channels:
\begin{align*}
\Phi_u^{\bar{v}, w}:B(H_u) \to B(H_w);  \quad\rho \mapsto \Tr_v(\alpha_u^{v,w}\rho(\alpha_u^{v,w})^*)\\
\Phi_u^{v, \bar{w}}:B(H_u) \to B(H_v);  \quad\rho \mapsto \Tr_w(\alpha_u^{v,w}\rho(\alpha_u^{v,w})^*).
\end{align*}
We name the above channels as Clebsch-Gordan channels (shortly, CG-channels) since the isometry $\alpha_u^{v,w}$ reflects the Clebsch-Gordan coefficients directly. Note that the symbol $\bar{v}$ does not refer to the conjugate representation, instead it means that we trace out the $H_v$ part. These channels have been studied by Al-Nuwairan \cite{Al14}, Brannan-Collins \cite{BrCo18b}, and also Leib-Solovej \cite{LiSo14}. It turns out that CG-channels preserve certain ``quantum symmetries''. Recall that groups provide a certain symmetry on quantum channels through their (projective) unitary representations, namely covariance of channels. This concept naturally extends to the case of quantum groups as follows.

\begin{defn}
Let $\Phi:B(H_A)\rightarrow B(H_B)$ be a quantum channel. Suppose that there are unitary representations $(u, H_A)$ and $(w, H_B)$ of a compact quantum group $\G$ such that
	\[(\iota \otimes \Phi)(\beta_u(\rho)) = \beta_w ( \Phi (\rho)), \qquad \rho \in B(H_A),\]
where $\beta_u$ and $\beta_w$ are $\G$-actions from \eqref{eq-G-action}.	
Then we say that the channel $\Phi$ is $\G$-covariant with respect to $(u,w)$. In case we have no possibility of confusion we simply say $\G$-covariant.
\end{defn}

Note that the covariance with respect to group representations has been studied in various contexts and has provided useful tools to handle information-theoretic problems \cite{Sc05, DaFuHo06, KoWe09, MeWo09, SaWoPeCi09, MaSp14, NaUe17, MoStDa17}.

We show that with mild assumptions, CG-channels are also {\it $\G$-covariant}.

\begin{prop}
Let $u$, $v$ and $w$ be irreducible unitary representations of a compact quantum group $\G$ such that $u \subset v \tp w$. Then the CG-channel $\Phi_u^{v, \bar{w}}$ is $\G$-covariant with respect to $(u,v)$ if the conjugate representation $w^c$ is also unitary. Similarly, $\Phi_u^{\bar{v}, w}$ is $\G^{\rm op}$-covariant with respect to $(u^*,w^*)$ if $v^c$ is unitary.
\end{prop}
\begin{proof}
We first check the case of $\Phi_u^{v, \bar{w}}$. For any quantum state $\rho \in B(H_u)$ we have
	\begin{align*}
		\lefteqn{(\Phi_u^{v, \bar{w}} \otimes \iota)(u (\rho \otimes 1)u^*)}\\
		& = \iota \otimes {\rm Tr} \otimes \iota [(\alpha^{v,w}_u \otimes \iota) u (\rho \otimes 1)u^* ((\alpha^{v,w}_u)^* \otimes \iota)]\\
		& = \iota \otimes {\rm Tr} \otimes \iota [ (v\tp w)(\alpha^{v,w}_u \otimes \iota) (\rho \otimes 1)((\alpha^{v,w}_u)^* \otimes \iota)(v\tp w)^* ]\\
		& = \sum^{d_v}_{i,j,i',j'=1}\sum^{d_w}_{k,l,k',l'=1} \iota \otimes {\rm Tr} [(|i \ra \la j| \otimes |k \ra \la l| )\alpha^{v,w}_u\rho (\alpha^{v,w}_u)^*(|j' \ra \la i'| \otimes |l' \ra \la k'| )] \otimes v_{ij}w_{kl}w^*_{k'l'}v^*_{i'j'}\\
		& = \sum^{d_v}_{i,j,i',j'=1}\sum^{d_w}_{l,l'=1} \iota \otimes {\rm Tr} [(|i \ra \la j| \otimes |l' \ra \la l| )\alpha^{v,w}_u\rho (\alpha^{v,w}_u)^*(|j' \ra \la i'| \otimes 1)] \otimes v_{ij}(\sum^{d_w}_{k=1}w_{kl}w^*_{k'l'})v^*_{i'j'}\\
		& = \sum^{d_v}_{i,j,i',j'=1}\sum^{d_w}_{l,l'=1} \iota \otimes {\rm Tr} [(|i \ra \la j| \otimes |l' \ra \la l| )\alpha^{v,w}_u\rho (\alpha^{v,w}_u)^*(|j' \ra \la i'| \otimes 1)] \otimes v_{ij}(w^tw^c)_{ll'}v^*_{i'j'}\\
		& = \sum^{d_v}_{i,j,i',j'=1}\iota \otimes {\rm Tr} [(|i \ra \la j| \otimes 1)\alpha^{v,w}_u\rho (\alpha^{v,w}_u)^*(|j' \ra \la i'| \otimes 1)] \otimes v_{ij}v^*_{i'j'}\\
		& = v (\Phi_u^{v, \bar{w}}(\rho) \otimes 1) v^*,
	\end{align*}
where we use tracial property for the fourth equality and the assumption that $w^c$ is unitary for $(w^tw^c)_{ll'} = \delta_{ll'}$.

For $\Phi_u^{\bar{v}, w}$ we observe that
	$$(\alpha^{v,w}_u \otimes \iota) u^*  (|\xi\ra \otimes 1 ) = (\alpha^{v,w}_u \otimes S) u  ( |\xi\ra \otimes 1)  = (\iota \otimes S)[(v\tp w) (\alpha^{v,w}_u  |\xi\ra \otimes 1) ],$$
where $S$ is the antipode of the quantum group $\G$. 
Thus, we get
	\begin{align*}
		\lefteqn{(\Phi_u^{\bar{v}, w} \otimes \iota)(u^* (\rho \otimes 1)u )}\\
		& = \sum^{d_v}_{i,j,i',j'=1}\sum^{d_w}_{k,l,k',l'=1} {\rm Tr} \otimes \iota [(|i \ra \la j| \otimes |k \ra \la l| )\alpha^{v,w}_u\rho (\alpha^{v,w}_u)^*(|j' \ra \la i'| \otimes |l' \ra \la k'| )] \otimes w^*_{lk}v^*_{ji}v_{j'i'}w_{l'k'}.
	\end{align*}
Then, we get the wanted conclusion by the same argument.
\end{proof}

The property $\G$-covariance has the following useful consequence.

\begin{prop}\label{prop-CGchannel-bistochastic}
Let $\Phi: B(H_u) \to B(H_v)$ be a quantum channel which is $\G$-covariant  with respect to a pair of unitary representations $(u,v)$ of a compact quantum group $\G$. If, in addition, $v$ is assumed to be irreducible, then $\Phi$ is bistochastic. In particular, all CG-channels associated to a Kac type compact quantum group are bistochastic.
\end{prop}
\begin{proof}
Since $\Phi$ is $\G$-covariant and $\frac{1_u}{d_u} \in \text{Hom}(u,u)$, we get $\Phi(\frac{1_u}{d_u}) \in \text{Hom}(v,v).$ But irreducibility and Schur's lemma then give $\Phi(\frac{1_u}{d_u})  \in \C I$, which implies $\Phi(\frac{1_u}{d_u})= \frac{1_v}{d_v}$.
\end{proof}

The following Proposition tells us that, under the assumption that $\g$ is of Kac type and $u\subseteq v \tp w$, the orthogonal projection from $H_v\otimes H_w$ onto $H_u$ can be obtained by applying an averaging technique using the Haar state, for each unit vector $\xi\in H_u$. Moreover, together with Theorem \ref{thm:choi-eq}, the following Proposition will be used to characterize EBT for TL-channels.

\begin{prop}\label{prop:ave}
Let $\g$ be a compact quantum group of Kac type and $u, v, w \in \Irr(\G)$ with $u \subset v \tp w$. Then for any unit vector $\xi\in \alpha^{v,w}_u(H_u)\subseteq H_v\otimes H_w $ we have
\[\frac{1}{d_u}\alpha^{v,w}_u (\alpha^{v,w}_u)^*= (\iota \otimes \iota\otimes h)((v \tp w)^* ( |\xi\ra\la \xi |\otimes 1 ) (v \tp w)) .\]
\end{prop}
\begin{proof}
Let $A=\displaystyle (  \iota \otimes \iota\otimes h)((v \tp w)^* (1\otimes |\xi\ra\la \xi | ) (v \tp w))$. Then, in order to reach the conclusion, it is enough to show that 
\[\la \eta | (\alpha^{v,w}_{u'})^*A \alpha^{v,w}_{u'} |\eta\ra=\frac{\delta_{u,u'}}{d_u} 1_u \]
for any irreducible components $u'$ of $v\tp w$ and any $\eta\in H_{u'}$. Indeed,

\begin{align*}
\la \eta | (\alpha^{v,w}_{u'})^*A \alpha^{v,w}_{u'} |\eta\ra&=h([ (\la \eta |( \alpha^{v,w}_{u'})^*  \otimes 1) (v\tp w)^* ] (|\xi \ra\la \xi |\otimes 1) [(v\tp w)(\alpha^{v,w}_{u'} |\eta \ra \otimes 1)] )\\
&=h( ( \la \eta |\otimes 1 )(u')^*  ( (\alpha^{v,w}_{u'})^* |\xi \ra\la \xi | \alpha^{v,w}_{u'}\otimes 1 )  u'( |\eta \ra \otimes 1)) .
\end{align*}

Then the facts that $\displaystyle ( \iota\otimes h)((u')^*(B\otimes 1) u')=\frac{\Tr(B)}{d_{u'}}1_{u'}$ and 
\[ \Tr((\alpha^{v,w}_{u'})^*|\xi \ra\la \xi | \alpha^{v,w}_{u'})= \la \xi | \alpha^{v,w}_{u'}(\alpha^{v,w}_{u'})^*|\xi\ra = \delta_{u,u'} \] 
complete the proof.

\end{proof}

\section{Temperley-Lieb Channels} \label{TL-diagrams}\label{sec:TL}

\subsection{Free orthogonal quantum groups $O_F^+$}

Let us fix an integer $N \ge 2$, $F \in \text{GL}_N(\C)$ satisfying $F \bar F = \pm  1$. We define $C(O_F^+)$ as the universal $C^*$-algebra generated by $u_{ij}$ $(1\leq i,j\leq N)$ with the defining relations (1) $\displaystyle u^*u= 1_N \otimes 1 = uu^*$ and (2) $u=(F\otimes 1) u^c (F^{-1}\otimes 1)$ where $u=(u_{ij})_{1\leq i,j\leq N}\in B(\Comp^N )\otimes C(O_F^+)$ that is called the {\it fundamental representation}. Then, together with a unital $*$-homomorphism $\Delta:C(O_F^+)\rightarrow C(O_F^+)\otimes_{\min}C(O_F^+)$ determined by 
$$\Delta(u_{ij})=\sum_{k=1}^N u_{ik}\otimes u_{kj},$$
$O_F^+=(C(O_F^+),\Delta)$ forms a compact quantum group, which is called the free orthogonal quantum group with parameter matrix $F$ \cite{VaWa96, Ba96, Ba97}. In particular, $O_F^+=SU(2)$ if $F=\left (\begin{array}{cc} 0&1\\ -1&0 \end{array} \right )$ and we denote by $O_N^+$ if $F=1_N$. Note that $O^+_F$ is of Kac type if and only if $F$ is unitary (\cite{Ba97}), which covers both of the above cases.

\subsection{Representations of $O_F^+$}

It is known from \cite{Ba96} that the irreducible representations of $O_F^+$ can be labelled  $(v^k)_{k \in \N_0}$ (up to unitary equivalence) in such a way that $v^0 = 1$, $v^1 = u$, the fundamental representation, $v^l \cong \overline{v^l}$, and the following fusion rule holds:
\begin{align} \label{frules}
v^l \tp v^m \cong \bigoplus_{0 \le r \le \min\{l,m\}} v^{l+m - 2r}.
\end{align}

Denote by $H_k$ the Hilbert space associated to $v^k$.  Then $H_0 = \C$, $H_1 = \C^N$, and \eqref{frules} shows that the dimensions $\dim H_k$ satisfy the recursion relations 
$\dim H_1 \dim H_k = \dim H_{k+1} + \dim H_{k-1}$.  Defining the quantum parameter \[q_0:= \frac{1}{N}\Big(\frac{2}{1+ \sqrt{1 -4/N^2}}\Big) \in (0,1],\] then one has $q_0 +q_0^{-1} = N$, and it can be shown by induction that the dimensions $\dim H_k$ are given by the {\it quantum integers} \[\dim H_k = [k+1]_{q_0}: = {q_0}^{-k}\Big(\frac{1-{q_0}^{2k+2}}{1-{q_0}^2}\Big) \qquad (N \ge 3).
 \]
When $N=2$, we have $q_0=1$, and then $\dim H_k = k+1 = \lim_{q_0 \to 1^-} [k+1]_{q_0}$.  Note that for $N \ge 3$, we have the exponential growth asymptotic $[k+1]_{q_0} \sim N^k$ (as $N \to \infty$).

We now describe the explicit construction of the representations $v^k$ and their corresponding Hilbert spaces $H_k$ due to Banica \cite{Ba96}.  (See also the description in \cite[Section 7]{VaVe07}).  The idea is that according to the fusion rules \eqref{frules}, the $k$-th tensor power $u^{\tiny \tp k}$ of the fundamental representation contains exactly one irreducible subrepresentation equivalent to $v^k$.  In particular, if we agree to explicitly identify $v^k$ as a subrepresentation of $u^{\tiny \tp k}$, then there exists a unique projection $0 \ne p_k \in \Hom_{O^+_F}(u^{{\tiny \tp} k}, u^{\tiny \tp k}) \subset B(H_1^{\otimes k})$ called the {\it Jones-Wenzl projection} \cite{Jo83,We87} satisfying $H_k = p_k(H_1^{\otimes k})$ and 
$$v^k = (p_k \otimes 1)u^{\tiny \tp k}(p_k \otimes 1) \in B(H_k)\otimes C(O_F^+).$$ 

Thus, we are left with the problem of describing the projection $p_k$.  To this end,   fix an orthonormal basis $(e_i)_{i=1}^N$ for $H_1= \C^N$, and put 
\begin{align} \label{cup}
\cup_F = \sum_{i=1}^N e_i \otimes Fe_i.
\end{align}  It is then a simple matter to check that $\cup_F  \in \text{Hom}_{O^+_F}(1,u\tp u)$, i.e. $ u^{\tiny \tp 2}(\cup_F \otimes 1) = (  \cup_F\otimes 1 ) $.  In particular, $\iota_{H_1^{\otimes i-1}} \otimes \cup_F\otimes  \iota_{H_1^{\otimes k-i-1}} \in \Hom_{O^+_F}(u^{\tiny \tp (k-2)}, u^{\tiny \tp k})$ for each $1 \le i \le k-1$.  Using these observations, we inductively define $(p_k)_{k \ge 1}$ using $p_1 = \iota_{H_1}$ together with the so-called {\it Wenzl recursion}
\begin{align}\label{Wenzl}
p_k = \iota_{H_1} \otimes p_{k-1} - \frac{[k-1]_{q}}{[k]_{q}}(\iota_{H_1}\otimes p_{k-1})(\cup_F \cup_F^* \otimes \iota_{H_1^{\otimes k-2}})(\iota_{H_1}\otimes p_{k-1}) \qquad (k \ge 2),
\end{align}
where $q = q(F) \in (0, q_0]$ is another quantum parameter defined so that $q+q^{-1} = \Tr(F^*F)$.

The Jones-Wenzl projections first appeared in the context of II$_1$-subfactors \cite{Jo83}.  The shared connection between subfactor theory and the representation theory of $O^+_F$ is through the famous {\it Temperley-Lieb category.}  Indeed, as explained for example in \cite{Ba96, BrCo18b, BrCo17b}, given $d \in (-\infty, 2] \cup [2, \infty)$ the Temperley-Lieb Category $\TL(d)$ is defined to be the strict C$^\ast$-tensor category generated by two simple objects $\{0,1\}$, where $0$ denotes the unit object for the tensor category, and $1 \ne 0$ is a self-dual simple object with the property that the morphism spaces $\TL_{k,l}(d):=\text{Hom}(1^{\otimes k}, 1^{\otimes l} )$ $(k,l \in \N)$ are generated by the identity map $ \iota \in \text{Hom}(1,1)$ together with a unique morphism $\cup\in \text{Hom}(0, 1 \otimes 1)$ satisfying $\cap\circ\cup = |d| \in \Hom(0,0) = \C$ and the ``snake equation'' $(\iota \otimes \cap)(\cup \otimes \iota)  = (\cap \otimes \iota)(\iota \otimes \cup) = \text{sgn}(d)\iota$.
Here, the ``cap'' $\cap$ is simply the adjoint $\cup^* \in \text{Hom}(1 \otimes 1,0)$ of the ``cup'' $\cup$.  On the other hand, we have the concrete C$^\ast$-tensor category $\text{Rep}(O^+_F)$ of finite dimensional unitary representations of $O^+_F$, and it was shown by Banica \cite{Ba96} that if $d = \Tr((F\bar F)(F^*F))$, then there exists a {\it unitary fiber functor} $\TL(d) \to \text{Rep}(O^+_F)$ which is  determined by mapping the simple objects $0,1 \in \TL(d)$ to $v^0, v^1 \in \text{Rep}(O^+_F)$, respectively, and by mapping the generating morphisms as follows
\[
\iota \in \TL_{1,1}(d) \mapsto \iota_{H_1} \in \text{Hom}_{O^+_F}(v^1, v^1)  \quad \& \quad  \cup \in \TL_{0,2}(d) \mapsto \cup_F \in \text{Hom}_{O^+_F}(v^0, v^1 \tp v^1).  
\]    
In other words, with $d$ and $F$ as above, we can concretely realize $\TL(d)$ in terms of the subcategory of finite dimensional Hilbert spaces $\text{Rep}(O^+_F)$.  In particular, for calculations involving morphisms and objects in $\text{Rep}(O^+_F)$, one can perform these calculations using the well-known planar diagrammatic calculus in the Temperley-Lieb category $\TL(d)$ \cite{BrCo17b, KaLi94, CaFlSa95}, which we now briefly review. 

\subsection{Diagrammatic calculus for $\text{Rep}(O^+_F)$}\label{subsec:Diagram}  In the following, we continue to use the notations (e.g. $H_k  = p_k(H_1^{\otimes k})$, $\cup_F$, etc.) defined above.  We use the standard string diagram calculus to depict linear transformations between Hilbert spaces.  That is, a linear operator $\rho \in B(H_k, H_l)$ will be diagrammatically represented as a string diagram
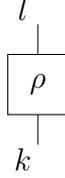
\[  \begin{tikzpicture}[baseline=(current  bounding  box.center),
			wh/.style={circle,draw=black,thick,inner sep=.5mm},
			bl/.style={circle,draw=black,fill=black,thick,inner sep=.5mm}, scale = 0.2]
		
\node  at (-1,9) {$l$};	
\node  at (-1,-1) {$k$};	
\node at (0,4) {$\rho$};

\draw [-, color=black]
		(0,6) -- (0,8);
		
\draw [-, color=black]
		(0,0) -- (0,2);		

\draw (-2,2) rectangle (2,6);
	\end{tikzpicture}\] 
with the input Hilbert space at the bottom of the diagram, and the output at the top.  The string corresponding to $H_l$ will be labeled by $l$.  We will generally omit the string corresponding to $H_0 = \C$, so a vector $\xi \in H_k \cong B(\C, H_k)$ and a covector $\xi^* \in H_k^* \cong B(H_k,\C)$ will be drawn, respectively, as

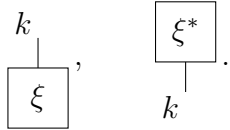
\[
  \begin{tikzpicture}[baseline=(current  bounding  box.center),
			wh/.style={circle,draw=black,thick,inner sep=.5mm},
			bl/.style={circle,draw=black,fill=black,thick,inner sep=.5mm}, scale = 0.2]
		
\node  at (-1,9) {$k$};	
\node at (0,4) {$\xi$};

\draw [-, color=black]
		(0,6) -- (0,8);
		
\draw (-2,2) rectangle (2,6);
	\end{tikzpicture}\, , \qquad   \begin{tikzpicture}[baseline=(current  bounding  box.center),
			wh/.style={circle,draw=black,thick,inner sep=.5mm},
			bl/.style={circle,draw=black,fill=black,thick,inner sep=.5mm}, scale = 0.2]
		
\node  at (-3,-1) {$k$};	
\node at (-2,4) {$\xi^*$};

\draw [-, color=black]
		(-2,0) -- (-2,2);		

\draw (-4,2) rectangle (0,6);
	\end{tikzpicture}\, .
\]
   Similarly, $\rho \in B (H_k \otimes H_l, H_{k'} \otimes H_{l'})$ is denoted using parallel input/output strings

 
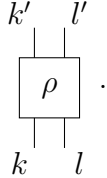
\[  \begin{tikzpicture}[baseline=(current  bounding  box.center),
			wh/.style={circle,draw=black,thick,inner sep=.5mm},
			bl/.style={circle,draw=black,fill=black,thick,inner sep=.5mm}, scale = 0.2]
		
\node  at (-2,9) {$k'$};	
\node  at (-2,-1) {$k$};	
\node  at (2,9) {$l'$};	
\node  at (2,-1) {$l$};
\node at (0,4) {$\rho$};

\draw [-, color=black]
		(-1,6) -- (-1,8);
		
\draw [-, color=black]
		(-1,0) -- (-1,2);

\draw [-, color=black]
		(1,6) -- (1,8);
		
\draw [-, color=black]
		(1,0) -- (1,2);			

\draw (-2,2) rectangle (2,6);
	\end{tikzpicture}.\] 
We define (for later use) the {\it ($k$-th) quantum trace}\footnote{The term ``trace'' comes from the fact that under the fiber functor $\TL(d) \to \text{Rep}(O^+_F)$, $\tau_k$ corresponds to the well-known Markov trace $\tau_k:\TL_{k,k}(d) \to \C$ obtained by tracial closure of Temperley-Lieb diagrams \cite{KaLi94}.} functional \[\tau_k:\mc B(H_1^{\otimes k}) \to \C, \qquad \tau_k(\rho) := \Tr_{H_1}^{\otimes k}((F^t\bar F )^{\otimes k})\rho) \qquad (k \in \N),\] which is depicted by the closure of a string diagram as follows: 

 \[  \begin{tikzpicture}[baseline=(current  bounding  box.center),
			wh/.style={circle,draw=black,thick,inner sep=.5mm},
			bl/.style={circle,draw=black,fill=black,thick,inner sep=.5mm}, scale = 0.2]
		
\node  at (-1,9) {$k$};	

\node at (0,4) {$\rho$};

\draw [-, color=black]
		(0,6) -- (0,8);
		
\draw [-, color=black]
		(0,0) -- (0,2);		

\draw (-1.5,2) rectangle (1.5,6);

\draw [-, color=black]
		(0,8) to [bend left = 90] (0,0);
	\end{tikzpicture} = \begin{tikzpicture}[baseline=(current  bounding  box.center),
			wh/.style={circle,draw=black,thick,inner sep=.5mm},
			bl/.style={circle,draw=black,fill=black,thick,inner sep=.5mm}, scale = 0.2]
		
\node  at (-1,9) {$k$};	

\node at (0,4) {$\rho$};

\draw [-, color=black]
		(0,6) -- (0,8);
		
\draw [-, color=black]
		(0,0) -- (0,2);		

\draw (-1.5,2) rectangle (1.5,6);

\draw [-, color=black]
		(0,8) to [bend right = 90] (0,0);
	\end{tikzpicture}\, .\] 
Composition of linear maps is depicted by vertical concatenation of string diagrams and tensoring is depicted by placing them in parallel, respectively.  

\[  \begin{tikzpicture}[baseline=(current  bounding  box.center),
			wh/.style={circle,draw=black,thick,inner sep=.5mm},
			bl/.style={circle,draw=black,fill=black,thick,inner sep=.5mm}, scale = 0.2]
		
\node  at (-1,9) {$l$};	
\node  at (-1,-1) {$k$};	
\node at (0,4) {$\rho \rho'$};

\draw [-, color=black]
		(0,6) -- (0,8);
		
\draw [-, color=black]
		(0,0) -- (0,2);		

\draw (-2,2) rectangle (2,6);
	\end{tikzpicture} = \begin{tikzpicture}[baseline=(current  bounding  box.center),
			wh/.style={circle,draw=black,thick,inner sep=.5mm},
			bl/.style={circle,draw=black,fill=black,thick,inner sep=.5mm}, scale = 0.2]

\node  at (-1,13) {$l$};	
\node  at (-1,-5) {$k$};	
		
\node at (0,8) {$\rho $};

\draw [-, color=black]
		(0,10) -- (0,12);
		
\draw [-, color=black]
		(0,4) -- (0,6);		

\draw (-2,6) rectangle (2,10);

\node at (0,0) {$ \rho'$};

\draw [-, color=black]
		(0,2) -- (0,4);
		
\draw [-, color=black]
		(0,-4) -- (0,-2);		

\draw (-2,-2) rectangle (2,2);
	\end{tikzpicture}, \qquad \begin{tikzpicture}[baseline=(current  bounding  box.center),
			wh/.style={circle,draw=black,thick,inner sep=.5mm},
			bl/.style={circle,draw=black,fill=black,thick,inner sep=.5mm}, scale = 0.2]
		
\node  at (-2,9) {$k'$};	
\node  at (-2,-1) {$k$};	
\node  at (2,9) {$l'$};	
\node  at (2,-1) {$l$};
\node at (0,4) {$\rho \otimes \rho'$};

\draw [-, color=black]
		(-1,6) -- (-1,8);
		
\draw [-, color=black]
		(-1,0) -- (-1,2);

\draw [-, color=black]
		(1,6) -- (1,8);
		
\draw [-, color=black]
		(1,0) -- (1,2);			

\draw (-3,2) rectangle (3,6);
	\end{tikzpicture} =  \begin{tikzpicture}[baseline=(current  bounding  box.center),
			wh/.style={circle,draw=black,thick,inner sep=.5mm},
			bl/.style={circle,draw=black,fill=black,thick,inner sep=.5mm}, scale = 0.2]
		
\node  at (-1,9) {$k'$};	
\node  at (-1,-1) {$k$};	
\node at (0,4) {$\rho$};

\draw [-, color=black]
		(0,6) -- (0,8);
		
\draw [-, color=black]
		(0,0) -- (0,2);		

\draw (-2,2) rectangle (2,6);
	\end{tikzpicture} \ \  \begin{tikzpicture}[baseline=(current  bounding  box.center),
			wh/.style={circle,draw=black,thick,inner sep=.5mm},
			bl/.style={circle,draw=black,fill=black,thick,inner sep=.5mm}, scale = 0.2]
		
\node  at (1,9) {$l'$};	
\node  at (1,-1) {$l$};	
\node at (0,4) {$\rho'$};

\draw [-, color=black]
		(0,6) -- (0,8);
		
\draw [-, color=black]
		(0,0) -- (0,2);		

\draw (-2,2) rectangle (2,6);
	\end{tikzpicture}
. \] 

Let us end this subsection by describing the string-diagrammatic representation of the maps specific to the representation category $\text{Rep}(O^+_F)$.  Recall that for $\text{Rep}(O^+_F)$, we have the fundamental generating morphisms $\iota_{H_k}$, $\cup_F$, $\cap_F := \cup_F^*$.  We depict these maps as follows:
\[
 \begin{tikzpicture}[baseline=(current  bounding  box.center),
			wh/.style={circle,draw=black,thick,inner sep=.5mm},
			bl/.style={circle,draw=black,fill=black,thick,inner sep=.5mm}, scale = 0.2]
		
\node  at (-1,9) {$k$};	
\node  at (-1,-1) {$k$};	
\node at (0,4) {$\iota_{H_k}$};

\draw [-, color=black]
		(0,6) -- (0,8);
		
\draw [-, color=black]
		(0,0) -- (0,2);		

\draw (-2,2) rectangle (2,6);
	\end{tikzpicture} = \   \begin{tikzpicture}[baseline=(current  bounding  box.center),
			wh/.style={circle,draw=black,thick,inner sep=.5mm},
			bl/.style={circle,draw=black,fill=black,thick,inner sep=.5mm}, scale = 0.2]

\node  at (-1,9) {$k$};	

\draw [-, color=black]
		(0,0) -- (0,8);

\end{tikzpicture} \ , \qquad \begin{tikzpicture}[baseline=(current  bounding  box.center),
			wh/.style={circle,draw=black,thick,inner sep=.5mm},
			bl/.style={circle,draw=black,fill=black,thick,inner sep=.5mm}, scale = 0.2]
		
\node  at (-2,9) {$1$};	
\node  at (2,9) {$1$};	
\node at (0,4) {$\cup_F$};

\draw [-, color=black]
		(-1,6) -- (-1,8);

\draw [-, color=black]
		(1,6) -- (1,8);

\draw (-2,2) rectangle (2,6);
	\end{tikzpicture} = \begin{tikzpicture}[baseline=(current  bounding  box.center),
			wh/.style={circle,draw=black,thick,inner sep=.5mm},
			bl/.style={circle,draw=black,fill=black,thick,inner sep=.5mm}, scale = 0.2]
		
\node  at (-2,9) {$1$};	
\node  at (2,9) {$1$};	

\draw [-, color=black]
		(-1,6) -- (-1,8);

\draw [-, color=black]
		(1,6) -- (1,8);

\draw[-, color=black]
(-1,6) to [bend right = 90] (1,6);

	\end{tikzpicture}, \qquad \begin{tikzpicture}[baseline=(current  bounding  box.center),
			wh/.style={circle,draw=black,thick,inner sep=.5mm},
			bl/.style={circle,draw=black,fill=black,thick,inner sep=.5mm}, scale = 0.2]
		
\node  at (-2,-1) {$1$};	

\node  at (2,-1) {$1$};

\node at (0,4) {$\cap_F$};

\draw [-, color=black]
		(-1,0) -- (-1,2);

\draw [-, color=black]
		(1,0) -- (1,2);			

\draw (-2,2) rectangle (2,6);
	\end{tikzpicture} = \begin{tikzpicture}[baseline=(current  bounding  box.center),
			wh/.style={circle,draw=black,thick,inner sep=.5mm},
			bl/.style={circle,draw=black,fill=black,thick,inner sep=.5mm}, scale = 0.2]
		
\node  at (-2,-1) {$1$};	

\node  at (2,-1) {$1$};

\draw [-, color=black]
		(-1,0) -- (-1,2);

\draw [-, color=black]
		(1,0) -- (1,2);	
		
\draw [-, color=black]
		(-1,2) to [bend left = 90](1,2);	

	\end{tikzpicture}
. \]
Then one has that the fundamental Temperley-Lieb relations are graphically depicted.  For example, the value of a closed loop is $|d|$:   
\[
\|\cup_F\|^2 = \cap_F \circ \cup_F = \begin{tikzpicture}[baseline=(current  bounding  box.center),
			wh/.style={circle,draw=black,thick,inner sep=.5mm},
			bl/.style={circle,draw=black,fill=black,thick,inner sep=.5mm}, scale = 0.2]

\node  at (2,-1) {$1$};

\draw [-, color=black]
		(-1,0) -- (-1,2);

\draw [-, color=black]
		(1,0) -- (1,2);	
		
\draw [-, color=black]
		(-1,2) to [bend left = 90](1,2);	

\draw [-, color=black]
		(-1,0) to [bend right = 90](1,0);	

	\end{tikzpicture} = \Tr(F^*F) = |d|, 
\]
and the snake equations are given by
\[
(\iota_{H_1} \otimes \cap_F) (\cup_F \otimes \iota_{H_1}) =  \begin{tikzpicture}[baseline=(current  bounding  box.center),
			wh/.style={circle,draw=black,thick,inner sep=.5mm},
			bl/.style={circle,draw=black,fill=black,thick,inner sep=.5mm}, scale = 0.2]

\draw [-, color=black]
		(0,0)--(0,3);

\draw [-, color=black]
		(2,0)--(2,3);

\draw [-, color=black]
		(4,0)--(4,3);

\draw [-, color=black]
		(2,3) to [bend left = 90](4,3);

\draw [-, color=black]
		(0,0) to [bend right = 90](2,0);

\end{tikzpicture} = F \bar F = \text{sgn}(d) \ \begin{tikzpicture}[baseline=(current  bounding  box.center),
			wh/.style={circle,draw=black,thick,inner sep=.5mm},
			bl/.style={circle,draw=black,fill=black,thick,inner sep=.5mm}, scale = 0.2]

\draw [-, color=black]
		(0,0)--(0,3);

\end{tikzpicture} \
=
 \begin{tikzpicture}[baseline=(current  bounding  box.center),
			wh/.style={circle,draw=black,thick,inner sep=.5mm},
			bl/.style={circle,draw=black,fill=black,thick,inner sep=.5mm}, scale = 0.2]

\draw [-, color=black]
		(0,0)--(0,3);

\draw [-, color=black]
		(2,0)--(2,3);

\draw [-, color=black]
		(4,0)--(4,3);

\draw [-, color=black]
		(0,3) to [bend left = 90](2,3);

\draw [-, color=black]
		(2,0) to [bend right = 90](4,0);

\end{tikzpicture} = (\cap_F \otimes \iota_{H_1}) (\iota_{H_1} \otimes \cup_F). 
\]

\subsection{Temperley-Lieb Channels}
We now come to our main objects of study, which are the CG-channels associated to the irreducible representations of the quantum groups $O^+_F$, which, in view of the above connection with the Temperley-Lieb category,  we redub ``Temperley-Lieb channels'':

\begin{defn}
A triple $(k,l,m) \in \N_0^3$ is called {\it admissible} if there exists an integer $0 \le r \le \min\{l,m\}$ such that $k = l+m - 2r$. For an admissible triple $(k,l,m) \in \N_0^3$ we have $v^k \subset v^l \tp v^m$ with the intertwining isometry $\alpha^{l,m}_k:H_k \to H_l \otimes H_m$ and the corresponding CG-channels $\Phi_{v^k}^{\overline{v^l}, v^m}$ and $\Phi_{v^k}^{v^l, \overline{v^m}}$ (shortly, $\Phi_k^{\bar{l}, m}$ and $\Phi_k^{l, \bar{m}}$)
are called {\it ($O^+_F$-)Temperley-Lieb channels}.
\end{defn}

Let us now give a string-diagrammatic description of the covariant isometries $\alpha_k^{l,m}$ which define the TL-channels above. We begin by fixing an admissible triple $(k,l,m) \in \N_0^3$ and define 
\begin{align}\label{unnormal}
A_k^{l,m} =(p_l \otimes p_m)\Big(\iota_{H_{l-r}} \otimes \cup_F^r  \otimes \iota_{m-r}\Big)p_k \in \text{Hom}_{O^+_F}(v^k, v^l \tp v^m).
\end{align} 
where $\cup_F^r \in \text{Hom}_{O^+_F}(v^0, v^{\tiny \tp 2r})$ is defined recursively from
	$$\cup_F^1 := \cup_F,\;\; \cup_F^r := (\iota_{H_1^{\otimes r-1}} \otimes \cup_F \otimes \iota_{H_1^{\otimes r-1}})\cup_F^{r-1}.$$
In terms of  our string diagram formalism, $\cup_F^r$ is given by $r$ nested cups 
\[
 \begin{tikzpicture}[baseline=(current  bounding  box.center),
			wh/.style={circle,draw=black,thick,inner sep=.5mm},
			bl/.style={circle,draw=black,fill=black,thick,inner sep=.5mm}, scale = 0.2]

\node at (0,4) {$\cup_F^r$};

\draw [-, color=black]
		(-1,6) -- (-1,8);
		
\draw [-, color=black]
		(-2,6) -- (-2,8);
		
\node at (0,7) {$...$};

\draw [-, color=black]
		(1,6) -- (1,8);
		
\draw [-, color=black]
		(2,6) -- (2,8);

\draw (-2,2) rectangle (2,6);

\end{tikzpicture} = 
\begin{tikzpicture}[baseline=(current  bounding  box.center),
			wh/.style={circle,draw=black,thick,inner sep=.5mm},
			bl/.style={circle,draw=black,fill=black,thick,inner sep=.5mm}, scale = 0.2]

\draw [-, color=black]
		(-1,6) -- (-1,8);
		
\draw [-, color=black]
		(-2,6) -- (-2,8);
		
\node at (0,7) {$...$};

\draw [-, color=black]
		(1,6) -- (1,8);
		
\draw [-, color=black]
		(2,6) -- (2,8);
		
\draw [-, color=black]
		(2,6) -- (2,8);

\draw [-, color=black]
		(2,6) -- (2,8);

\draw [-, color=black]
		(-2,6) to [bend right = 90] (2,6);

\draw [-, color=black]
		(-1,6) to [bend right = 90] (1,6);

\end{tikzpicture},
\]
and $A_k^{l,m}$ is given by  \[
 A_k^{l,m} = \begin{tikzpicture}[baseline=(current  bounding  box.center),
			wh/.style={circle,draw=black,thick,inner sep=.5mm},
			bl/.style={circle,draw=black,fill=black,thick,inner sep=.5mm}, scale = 0.2]

\draw (-3,0) rectangle (3,2);
\node at (0,1) {$p_k$};

\draw (-9,10) rectangle (-4,12);
\node at (-6.5,11) {$p_l$};

\draw (4,10) rectangle (9,12);
\node at (6.5,11) {$p_m$};

\draw [-, color=black]
		(-3,2) -- (-9,10);

\draw [-, color=black]
		(-1,2) -- (-7,10);

\draw [-, color=black]
		(3,2) -- (9,10);

\draw [-, color=black]
		(1,2) -- (7,10);

\draw [-, color=black]
		(-4,10) to [bend right = 40] (4,10);

\draw [-, color=black]
		(-6,10) to [bend right = 90] (6,10);

\node at (0,8) {$\vdots$};
\node at (-4.3,5) {$\cdot \cdot$};
\node at (4.3,5) {$\cdot \cdot$};

\end{tikzpicture}
\]

The (non-zero) map $A_{k}^{l,m}$ is often called a {\it three-vertex} in the context of tensor category theory and Temperley-Lieb recoupling theory  \cite{KaLi94}, and (following standard conventions) the above string diagram for $A_k^{l,m}$ is simply drawn as a trivalent vertex:

\begin{center}

	\begin{tikzpicture}[baseline=(current  bounding  box.center),
			wh/.style={circle,draw=black,thick,inner sep=.5mm},
			bl/.style={circle,draw=black,fill=black,thick,inner sep=.5mm}, scale = 0.2]

\node  at (-5,0.5) {$A_{k}^{l,m} = $};	
\node  at (-5,5) {$l$};	
\node  at (5,5) {$m$};	
\node  at (-1,-5) {$k$};	

\draw [-, color=black]
		(-4,4) -- (0,0);

\draw [-, color=black]
		(0,-0) -- (4,4);

\draw [-, color=black]
		(0,-4) -- (0,0);

	\end{tikzpicture}.
 
\end{center}

We then have that the the adjoint $(A_k^{l,m})^*$ is obtained by rotating 180 degrees about the horizontal axis. 
\[
\begin{tikzpicture}[baseline=(current  bounding  box.center),
			wh/.style={circle,draw=black,thick,inner sep=.5mm},
			bl/.style={circle,draw=black,fill=black,thick,inner sep=.5mm}, scale = 0.2]

\node  at (-5,0.5) {$(A_{k}^{l,m})^* = \quad $};	
\node  at (-5,-5) {$l$};	
\node  at (5,-5) {$m$};	
\node  at (-1,5) {$k$};	

\draw [-, color=black]
		(-4,-4) -- (0,0);

\draw [-, color=black]
		(0,-0) -- (4,-4);

\draw [-, color=black]
		(0,4) -- (0,0);

	\end{tikzpicture}.
\]
From Schur's Lemma and irreducibility, it follows that our required isometry $\alpha_k^{l,m}$ must be a scalar multiple of the three-vertex $A_k^{l,m}$, and this scaling factor is given in terms of the so-called {\it theta-net}  $\theta_q(k,l,m)$ \cite{KaLi94}.
\[
\theta_q(k,l,m):= \tau_k((A_k^{l,m})^*A_k^{l,m}) = \begin{tikzpicture}[baseline=(current  bounding  box.center),
			wh/.style={circle,draw=black,thick,inner sep=.5mm},
			bl/.style={circle,draw=black,fill=black,thick,inner sep=.5mm}, scale = 0.17]

\node  at (-4.5,5) {{\scriptsize $l$}};	
\node  at (1.5,5) {{\scriptsize $m$}};	
\node  at (-1,-5) {{\scriptsize $k$}};	
\node at (-1,13) {{\scriptsize $k$}};

\draw [-, color=black]
		(-4,4) -- (0,0);

\draw [-, color=black]
		(0,-0) -- (4,4);

\draw [-, color=black]
		(0,-4) -- (0,0);
		
\draw [-, color=black]
		(-4,4) -- (0,8);
		
\draw [-, color=black]
		(4,4) -- (0,8);
		
\draw [-, color=black]
		(0,8) -- (0,12);		

\draw [-, color=black]
		(0,12) to [bend left =110] (0,-4);

	\end{tikzpicture} = \frac{[r]_q![l-r]_q![m-r]_q![k+r+1]_q!}{[l]_q![m]_q![k]_q!},
\]
where $q=q(F)$, $k=l+m - 2r$, and $[x]_q! = [x]_q[x-1]_q \ldots [2]_q[1]_q$ denotes the quantum factorial.  Then one has
\begin{center}

	\begin{tikzpicture}[baseline=(current  bounding  box.center),
			wh/.style={circle,draw=black,thick,inner sep=.5mm},
			bl/.style={circle,draw=black,fill=black,thick,inner sep=.5mm}, scale = 0.2]

\node  at (-22,0.5) {$\alpha_k^{l,m} = \Big(\frac{\tau_k(\iota_{H_k})}{ \tau_k((A_k^{l,m})^*A_k^{l,m})}\Big)^{1/2} A_k^{l,m} =  \Big(\frac{[k+1]_q}{\theta_q(k,l,m)}\Big)^{1/2}\ \ \  \ $};	
\node  at (-5,5) {$l$};	
\node  at (5,5) {$m$};	
\node  at (-1,-5) {$k$};	

\draw [-, color=black]
		(-4,4) -- (0,0);

\draw [-, color=black]
		(0,-0) -- (4,4);

\draw [-, color=black]
		(0,-4) -- (0,0);

	\end{tikzpicture}.

\end{center}

\subsection{Kac type Temperley-Lieb channels}
Throughout the rest of the paper we make the standing assumption that all free orthogonal quantum groups $O^+_F$ under consideration are of Kac type, which is equivalent to the unitarity of $F$ \cite{Ba97}. (In fact, for the most part we just consider $O^+_N$, however this slightly higher level of generality is useful at times, allowing us for exmple to prove results for $SU(2)$ simultaneously). The main reason for making this assumption is that for the calculations that follow, it is essential for us to have that the ``physical operations'' of taking partial traces in tensor product spaces such as $B(H_l \otimes H_m)$ agree with the ``quantum operations'' coming from taking (partial) quantum traces using the functionals  $\tau_k$ described above. In this case, we also have the handy feature that the $O^+_F$-covariant unit vectors $\alpha_0^{k,k}\in H_k \otimes H_k$ are all maximally entangled states. 

\begin{rem}
Note that when $O^+_F$ is of Kac type, we have that both the quantum parameters $q_0$ and $q$ defined above are equal (since $N = \Tr(F^*F)$ when $F$ is unitary).   From now on we simply use the letter $q$ to denote the quantum parameter.
\end{rem}

Of course, since in the Kac case the quantum traces and ordinary traces agree, we have the following diagrammatic representations for the Temperley-Lieb quantum channels $ \Phi_k^{\bar l, m},  \Phi_k^{l, \bar m}$:
\[ \Phi_k^{\bar l, m}(\rho)=\frac{[k+1]_q}{\theta_q(k,l,m)} \begin{tikzpicture}[baseline=(current  bounding  box.center),
			wh/.style={circle,draw=black,thick,inner sep=.5mm},
			bl/.style={circle,draw=black,fill=black,thick,inner sep=.5mm}, scale = 0.17]
\node  at (5,15) {$m$};	
\node  at (5,-7) {$m$};	
\node  at (-9,5) {$l$};	

\node  at (-1,9) {$k$};	
\node  at (-1,-1) {$k$};	
\node at (0,4) {$\rho$};

\draw [-, color=black]
		(-4,16) -- (0,12);

\draw [-, color=black]
		(0,12) -- (4,16);

\draw [-, color=black]
		(0,6) -- (0,12);
		
\draw [-, color=black]
		(-4,-8) -- (0,-4);
		
\draw [-, color=black]
		(4,-8) -- (0,-4);
		
\draw [-, color=black]
		(0,-4) -- (0,2);		

\draw [-, color=black]
		(-4,16) to [bend right=120] (-4,-8);

\draw (-2,2) rectangle (2,6);
	\end{tikzpicture} \quad, \quad  \Phi_k^{l, \overline{m}}(\rho)=\frac{[k+1]_q}{\theta_q(k,l,m)}\begin{tikzpicture}[baseline=(current  bounding  box.center),
			wh/.style={circle,draw=black,thick,inner sep=.5mm},
			bl/.style={circle,draw=black,fill=black,thick,inner sep=.5mm}, scale = 0.17]
\node  at (-5,15) {$l$};	
\node  at (-5,-7) {$l$};	
\node  at (8.5,5) {$m$};
\node  at (-1,9) {$k$};	
\node  at (-1,-1) {$k$};	
\node at (0,4) {$\rho$};

\draw [-, color=black]
		(-4,16) -- (0,12);

\draw [-, color=black]
		(0,12) -- (4,16);

\draw [-, color=black]
		(0,6) -- (0,12);
		
\draw [-, color=black]
		(-4,-8) -- (0,-4);
		
\draw [-, color=black]
		(4,-8) -- (0,-4);
		
\draw [-, color=black]
		(0,-4) -- (0,2);		

\draw [-, color=black]
		(4,16) to [bend left=120] (4,-8);

\draw (-2,2) rectangle (2,6);
	\end{tikzpicture}.\] 

Let us finish this section with an application of our string diagram formalism to the Choi maps associated to the TL-channels. The result below was proved for the cases of $SU(2)$ by Al-Nuwairan \cite{Al14} and $O^+_N$ in \cite{BrCo17b}.  The following general case follows by the exact same planar isotopy arguments used in \cite{BrCo17b}.

\begin{thm} \label{thm:choi-eq}
For any admissible triple $(k,l,m)\in \n^3_0$ the Choi matrices associated to any Kac type $O^+_F$-TL-channels $\Phi_k^{\bar{l}, m}$ and $\Phi_k^{l, \bar{m}}$ are given by
\begin{equation}\label{choi-eq}
C_{\Phi_k^{\overline{l}, m}} = \frac{[k+1]_q}{[l+1]_q} \alpha_l^{m,k}(\alpha_{l}^{m,k})^*,\;\;
C_{\Phi_k^{l, \overline{m}}} = \frac{[k+1]_q}{[m+1]_q} \alpha_m^{k,l}(\alpha_{m}^{k,l})^*,
\end{equation}
respectively. In particular, these Choi maps are scalar multiples of $O^+_F$-covariant projections onto irreducible subrepresentations.
\end{thm}

\section{The minimum output entropy and capacities of $O_N^+$-Temperley-Lieb channels}\label{sec:moe-cap}

In this section we establish asymptotically sharp estimates on the minimum output entropy, the Holevo capacity and the ``one-shot'' quantum capacity of $O_N^+$-TL-channels for large enough $N$. The estimate begins with the following result of \cite[Corollary 4.2]{BrCo18b}.
	\begin{equation}
		H_{\min}(\Phi^{l,\bar m}_k )=H_{\min}(\Phi^{\overline{l},m}_k)\geq \log(\frac{\theta_q(k,l,m)}{[k+1]_q})\ge \frac{l+m-k}{2} \cdot \log N - C(N)
	\end{equation}
with $C(N)\to 0$ as $N\to \infty$. The above estimate was conjectured to be asymptotically optimal as $N\to \infty$ in \cite{BrCo18b}, which will be confirmed to be true below.

Before we dig into the above conjecture we prepare several elementary estimates. Let $f(t) = -t\log t$, $0<t<1$ be the function we use for the entropy. Then it is straightforwad to see that $f(t) \lesssim t^{1/2}$ and $f(t) \lesssim 1-t$, where $a\lesssim b$ means that there is a universal constant $C>0$ such that $a \le C \cdot b$. The Fannes-Audenaert inequality (\cite{A07}) says that for any quantum states $X,Y \in B(H)$ with ${\rm dim}H = n$
	$$|H(X) - H(Y)| \le \delta \log (n-1) + f(\delta) + f(1-\delta),\; \delta = \frac{1}{2}||X-Y||_1,$$
where $||\cdot||_1$ is the trace norm, so that we have
	\begin{equation}
	|H(X) - H(Y)| \lesssim \log n  \cdot ||X-Y||_1 + ||X-Y||^{1/2}_1.
	\end{equation}	

\begin{lem}\label{lem-Fannes-extend}
Let $X,Y \in B(H)_+$ with ${\rm dim}H = n$. Suppose further that $\Tr(X) = 1 \ge \Tr(Y) > 0$. Then we still have
	\begin{equation}
	|H(X) - H(Y)| \lesssim \log n  \cdot ||X-Y||_1 + ||X-Y||^{1/2}_1.
	\end{equation}
\end{lem}
\begin{proof}
First we observe that
	\begin{align*}
	H(X) - H(Y)
	& = H(X) + \Tr(Y) \log \Tr(Y) - \Tr(Y)H(\frac{Y}{\Tr(Y)})\\
	& = \Tr(Y) \log \Tr(Y) + (1-\Tr(Y))H(X) + \Tr(Y) (H(X) - H(\frac{Y}{\Tr(Y)}))\\
	& = A + B + C.
	\end{align*}
Since we have $1- \Tr(Y) = \Tr(X-Y) \le ||X-Y||_1$	we know
	$$|A| \lesssim ||X-Y||_1,\; |B| \lesssim \log n  \cdot ||X-Y||_1.$$
For the third term we have
	$$|C| \le |H(X) - H(\frac{Y}{\Tr(Y)})| \lesssim \log n  \cdot ||X-\frac{Y}{\Tr(Y)}||_1 + ||X-\frac{Y}{\Tr(Y)}||^{1/2}_1.$$
Finally we observe that
	$$||X-\frac{Y}{\Tr(Y)}||_1 \le ||X-Y||_1 + (\frac{1}{\Tr(Y)} - 1)||Y||_1 = ||X-Y||_1 + 1 - \Tr(Y) \le 2||X-Y||_1,$$
which leads us to the conclusion we wanted.	
\end{proof}

\begin{lem}\label{lem-factorial}
For any admissible $(l,m,k)\in \n^3_0$ with $k=l+m-2r$ we have
\[ \frac{N^r[k+1]_q}{\theta_q(k,l,m)}=1+O(\frac{1}{N^2}).\]
\end{lem}
\begin{proof}
We first observe for any $k\ge 1$ that
	\begin{align*}
		\frac{[k+1]_q}{[k]_qN}
		& = \frac{1}{2}(1 + \sqrt{1-4/N^2}) \frac{1-q^{2k+2}}{1-q^{2k}}\\
		& = \frac{1}{2}(1 + \sqrt{1-4/N^2}) (1 + \frac{q^{2k}-q^{2k+2}}{1-q^{2k}})\\
		& = \frac{1}{2}(2  + O(\frac{1}{N^2})) (1+O(\frac{1}{N^2})) = 1+O(\frac{1}{N^2}).
	\end{align*}
Then, we can easily see for all $a > b\in \n$ that
	$$\frac{[a]_q}{[b]_q N^{a-b}} = \frac{[a]_q}{[a-1]_q N}\cdots \frac{[b+1]_q}{[b]_q N} = (1+O(\frac{1}{N^2}))^{a-b} = 1+O(\frac{1}{N^2}),$$
which can be extended to the following
	\[\frac{[a]_q!}{[b]_q![a-b]_q! N^{b(a-b)}}=1+O(\frac{1}{N^2})=\frac{[b]_q! [a-b]_q! N^{b(a-b)}}{[a]_q!}.\]
Finally, we have
\begin{align*}
\frac{N^r[k+1]_q}{\theta_q(k,l,m)}
& = N^r\frac{[l]_q! [m]_q! [k+1]_q!}{[r]_q![l-r]_q![m-r]_q! [k+r+1]_q!}\\
& = \frac{[l]_q!}{[r]_q![l-r]_q!N^{r(l-r)}}\cdot \frac{[m]_q!}{[m-r]_q![r]_q!N^{r(m-r)}}\cdot \frac{[k+1]_q![r]_q!N^{r(k+1)}}{[k+r+1]_q!}\\
& = 1+O(\frac{1}{N^2})
\end{align*}
since $-r(l-r)-r(m-r)+r(k+1)=r(2r-l-m+k+1)=r$.

\end{proof}

Here, we introduce some notations. For $N \ge 2$ we write the index set $I = \{1, 2, \cdots, N\}$. We also need multi-index sets	
	$$I^n = \{ {\bf i} = (i_1, \cdots, i_n): i_k \in I,\; 1\le k \le n\}$$
and	
	$$I^n_{\ne} := \{ {\bf i} = (i_1, \cdots, i_n) \in I^n: i_k \ne i_{k+1}, \;1\le k \le n-1\}.$$
We sometimes need to aviod particular indices as follows.	
	$$(s,t)/I^n_{\ne} := \{ {\bf i} = (i_1, \cdots, i_n) \in I^n_{\ne}: i_1 \ne s, i_1 \ne t\}$$
and
	$$I^n_{\ne}\backslash(t) := \{ {\bf i} = (i_1, \cdots, i_n) \in I^n_{\ne}: i_n \ne t\}$$		
for $n\in \n$, $s \ne t\in I$. Note that we have $|(s,t)/I^n_{\ne}| = (N-2)(N-1)^{n-1}$ and $|I^n_{\ne}\backslash(t)| = (N-1)^n$.

For each ${\bf i} \in  I^n_{\neq}$ we can easily see that $|{\bf i}\ra \in H_n$ so that $p_n |{\bf i}\ra = |{\bf i}\ra$ from the Jones-Wenzl recursion.

For ${\bf i} \in I^n$ and ${\bf j} \in I^m$ the vector $| {\bf i} \ra \otimes | {\bf j} \ra \in \Comp^{n+m}$ will simply be denoted by $| {\bf i} {\bf j} \ra$. We will use a very specific index ${\bf m}^k := (1,2,1,\cdots) \in I^k$, $k\ge 1$. For ${\bf i}=(i_1,\cdots,i_n)\in I^n$, its order reversed multi-index $\check{{\bf i}}=(i_n,\cdots, i_1)\in I^n$ will be considered.

\begin{thm}\label{thm-MOE-sharp}
For each admissible triple $(l,m,k)\in \n^3_0$ we have
	\begin{equation}
		\frac{l+m-k}{2} \cdot \log N - C(N) \le H_{\min}(\Phi^{l,\bar m}_k )=H_{\min}(\Phi^{\overline{l},m}_k)\le \frac{l+m-k}{2} \cdot \log N + D(N)
	\end{equation}
with $C(N), D(N)\to 0$ as $N\to \infty$. When $k=l+m$, we actually have the following.
	$$H_{\min}(\Phi^{l,\bar m}_{l+m})=H_{\min}(\Phi^{\bar{l},m}_{l+m}) = 0$$
for any $N\ge 2$.	
\end{thm}
\begin{proof}
We set $r = \frac{l+m-k}{2}$. We will use a very specific index ${\bf m} := (1,2,1,\cdots) \in H_k\subseteq H^{\otimes k}_1$, which splits into $(m_1,\cdots, m_k) = {\bf m} = {\bf m}'{\bf m}''$, where ${\bf m}' = (m_1,\cdots, m_{l-r}) \in H_{l-r}\subseteq H^{\otimes l-r}_1$ and ${\bf m}'' = (m_{l-r+1},\cdots, m_k)\in H_{m-r}\subseteq H^{\otimes m-r}_1$. Then, we have
	\begin{align*}
		\frac{\theta_q(k,l,m)}{[k+1]_q} \Phi^{\bar{l},m}_k(|{\bf m}\ra \la{\bf m}|)
		& = \Tr \otimes \iota (A^{l,m}_k |{\bf m}\ra \la{\bf m}| (A^{l,m}_k )^*)\\
		& = \Tr \otimes \iota (A^{l,m}_k |{\bf m}'{\bf m}''\ra \la{\bf m}'{\bf m}''| (A^{l,m}_k )^*)\\
		& = \sum_{{\bf i}, {\bf i}' \in I^r} \Tr \otimes \iota [(p_l \otimes p_m) (|{\bf m}'{\bf i} \ra  \la{\bf m}'{\bf i}'|  \otimes |\check{{\bf i}}\, {\bf m}''\ra\la \check{{\bf i}'}{\bf m}''|) (p_l \otimes p_m)]\\
		& = \sum_{{\bf i}, {\bf i}' \in I^r} \la{\bf m}'{\bf i}'| p_l |{\bf m}'{\bf i} \ra \cdot p_m |\check{{\bf i}}\, {\bf m}''\ra\la \check{{\bf i}'}{\bf m}''| p_m\\
		& = \sum_{{\bf i}\in (1,2)/I^r_{\ne}} |\check{{\bf i}}\, {\bf m}''\ra\la \check{{\bf i}}{\bf m}''| + \sum_{{\bf i}, {\bf i}' \not\in (1,2)/I^r_{\ne}} \la{\bf m}'{\bf i}'| p_l |{\bf m}'{\bf i} \ra \cdot p_m |\check{{\bf i}}\, {\bf m}''\ra\la \check{{\bf i}'}{\bf m}''| p_m\\
		& = \frac{\theta_q(k,l,m)}{[k+1]_q}(Z(1) + Z(2)),
	\end{align*} 
where we used the fact that for ${\bf i}\in (1,2)/I^r_{\ne}$ we have ${\bf m}'{\bf i}\in H_l$ and $\check{{\bf i}}\, {\bf m}'' \in H_m$. Note that 
$$\frac{\theta_q(k,l,m)}{[k+1]_q}Z(2)= \mathrm{Tr}\otimes \iota ((p_l\otimes p_m)|\xi \ra \la \xi | (p_l\otimes p_m))\geq 0, $$
where $|\xi\ra = \displaystyle \sum_{{\bf i}\notin (1,2)/I^r_{\neq}} |{\bf m}' {\bf i}\ra\otimes |\check{{\bf i}}{\bf m}''\ra$. 
The term $Z(1)$ is the dominant one with the entropy
	\begin{align*}
	H(Z(1))
	& = (N-2)(N-1)^{r-1}\frac{[k+1]_q}{\theta_q(k,l,m)} \log \frac{\theta_q(k,l,m)}{[k+1]_q}\\
	& = (1-\frac{2}{N})(1-\frac{1}{N})^{r-1}\frac{N^r[k+1]_q}{\theta_q(k,l,m)} \log \frac{\theta_q(k,l,m)}{[k+1]_q}\\
	& = (1+O(\frac{1}{N}))\log[ (1+O(\frac{1}{N^2}))N^r]
	\end{align*}
by Lemma \ref{lem-factorial}. For the second term $Z(2)$ we have
	$$	\Tr(Z(2))
	= 1 - \Tr(Z(1))\\
	= 1- (N-2)(N-1)^{r-1}\frac{[k+1]_q}{\theta_q(k,l,m)}\\
	= O(\frac{1}{N}).$$
By Lemma \ref{lem-Fannes-extend} we have
	$$| H(\Phi^{\bar{l},m}_k(|{\bf m}\ra \la{\bf m}|)) - H(Z(1)) | \lesssim m \log N \Tr(Z(2)) + \Tr(Z(2))^{1/2} \lesssim O(\frac{1}{\sqrt{N}}),$$
which leads us to the conclusion we wanted.

If $k=l+m$, then we have $r=0$ and
	$$\Phi^{\bar{l},m}_{l+m}(|{\bf m}\ra \la{\bf m}|) = |{\bf m}'\ra \la{\bf m}'|,$$
which is a pure state. Thus, we get the conclusion we wanted.	
\end{proof}

Now we move to the case of capacities. We will apply a similar argument for the lower bound of ``one-shot'' quantum capacity.
\begin{thm}\label{thm-QC}
For each admissible triple $(k,l,m)\in \n^3_0$ we have
	\begin{equation}
		\begin{cases}\frac{l+k-m}{2} \cdot \log N - C(N) \le Q^{(1)}(\Phi^{l,\bar m}_k)\\ \frac{m+k-l}{2} \cdot \log N - D(N) \le Q^{(1)}(\Phi^{\overline{l},m}_k)\end{cases}
	\end{equation}
with constants $C(N), D(N)\to 0$ as $N\to \infty$. When $k=l+m$, we actually have the following.
	\begin{equation}
		\begin{cases} l \cdot \log (N-1) \le Q^{(1)}(\Phi^{l,\bar m}_{l+m})\\ m \cdot \log (N-1) \le Q^{(1)}(\Phi^{\overline{l},m}_{l+m}).\end{cases}
	\end{equation}
\end{thm}
\begin{proof}
We set $r = \frac{l+m-k}{2}$ and fix a specific index ${\bf n} := (1,2,1,\cdots) \in H_{m-r}\subseteq H^{\otimes m-r}_1$. We first consider the estimates of $Q^{(1)}(\Phi^{\bar{l},m}_k)$. For any ${\bf j} \in I^{l-r}_{\ne}\backslash(1)$ we use the same argument as in the proof of Theorem \ref{thm-MOE-sharp} to get
	\begin{align*}
		\frac{\theta_q(k,l,m)}{[k+1]_q} \Phi^{\bar{l},m}_k(|{\bf j}{\bf n}\ra \la {\bf j}{\bf n}|)
		& = \Tr \otimes \iota (A^{l,m}_k |{\bf j}{\bf n}\ra \la{\bf j}{\bf n}| (A^{l,m}_k )^*)\\
		& = \sum_{{\bf i}, {\bf i}' \in I^r} \la{\bf j}{\bf i}'| p_l |{\bf j}{\bf i} \ra \cdot p_m |\check{{\bf i}}\, {\bf n}\ra\la \check{{\bf i}'}{\bf n}| p_m\\
		& = \sum_{{\bf i}\in (1,j_{l-r})/I^r_{\ne}} |\check{{\bf i}}{\bf n}\ra\la \check{{\bf i}}{\bf n}| + \sum_{{\bf i}, {\bf i}' \not\in (1,j_{l-r})/I^r_{\ne}} \la{\bf j}{\bf i}'| p_l |{\bf j}{\bf i} \ra \cdot p_m |\check{{\bf i}}\, {\bf n}\ra\la \check{{\bf i}'}{\bf n}| p_m\\
		& = \frac{\theta_q(k,l,m)}{[k+1]_q}(Z(1,{\bf j}) + Z(2,{\bf j})).
	\end{align*}
Now we set $\displaystyle \rho = \frac{1}{(N-1)^{l-r}}\sum_{{\bf j} \in I^{l-r}_{\ne}\backslash(1)}|{\bf j}{\bf n}\ra \la {\bf j}{\bf n}|$ and we get
		$$\Phi^{\bar{l},m}_k(\rho) = \frac{1}{(N-1)^{l-r}}\sum_{{\bf j} \in I^{l-r}_{\ne}\backslash(1)}(Z(1,{\bf j}) + Z(2,{\bf j})) = Z(1) + Z(2).$$
In other words,
	\begin{align*}
	Z(1)
	& = \frac{[k+1]_q}{(N-1)^{l-r}\theta_q(k,l,m)}\sum_{{\bf j} \in I^{l-r}_{\ne}\backslash(1)}\sum_{{\bf i}\in (1,j_{l-r})/I^r_{\ne}} |\check{{\bf i}}{\bf n}\ra\la \check{{\bf i}}{\bf n}|\\
	& = \frac{[k+1]_q}{(N-1)\theta_q(k,l,m)}\sum^N_{j_{l-r} =2 }\sum_{{\bf i}\in (1,j_{l-r})/I^r_{\ne}} |\check{{\bf i}}{\bf n}\ra\la \check{{\bf i}}{\bf n}|\\
	& = \frac{(N-2)[k+1]_q}{(N-1)\theta_q(k,l,m)}\sum_{{\bf i}\in (1)/I^r_{\ne}} |\check{{\bf i}}{\bf n}\ra\la \check{{\bf i}}{\bf n}|.
	\end{align*}
As before we use Lemma \ref{lem-factorial} to get
	\begin{align*}
	H(Z(1))
	& = (N-1)^r \frac{(N-2)[k+1]_q}{(N-1)\theta_q(k,l,m)} \log \frac{(N-1)\theta_q(k,l,m)}{(N-2)[k+1]_q}\\
	& = (1-\frac{1}{N})^r \frac{N-2}{N-1}\frac{N^r[k+1]_q}{\theta_q(k,l,m)} \log \frac{(N-1)\theta_q(k,l,m)}{(N-2)[k+1]_q}\\
	& = (1+O(\frac{1}{N}))\log[ (1+O(\frac{1}{N}))N^r]
	\end{align*}
and
	$$\Tr(Z(2)) = 1 - \Tr(Z(1)) = 1 - (N-1)^r \frac{(N-2)[k+1]_q}{(N-1)\theta_q(k,l,m)} = O(\frac{1}{N}).$$
By Lemma \ref{lem-Fannes-extend} again we still have
	$$| H(\Phi^{\bar{l},m}_k(\rho)) - H(Z(1)) | \lesssim m \log N \Tr(Z(2)) + \Tr(Z(2))^{1/2} \lesssim O(\frac{1}{\sqrt{N}}).$$
			
For the complementary channel we similarly have
	\begin{align*}
		\frac{\theta_q(k,l,m)}{[k+1]_q} \Phi^{l,\bar{m}}_k(|{\bf j}{\bf n}\ra \la {\bf j}{\bf n}|)
		& = \iota \otimes \Tr (A^{l,m}_k |{\bf j}{\bf n}\ra \la{\bf j}{\bf n}| (A^{l,m}_k )^*)\\
		& = \sum_{{\bf i}, {\bf i}' \in I^r} p_l |{\bf j}{\bf i} \ra \la{\bf j}{\bf i}'| p_l \cdot \la \check{{\bf i}'}{\bf n}| p_m |\check{{\bf i}}\, {\bf n}\ra\\
		& = \sum_{{\bf i}\in (1,j_{l-r})/I^r_{\ne}} |{\bf j}{\bf i} \ra \la{\bf j}{\bf i}| + \sum_{{\bf i}, {\bf i}' \not\in (1,j_{l-r})/I^r_{\ne}}  \la \check{{\bf i}'}{\bf n}| p_m |\check{{\bf i}}\, {\bf n}\ra \cdot p_l |{\bf j}{\bf i} \ra \la{\bf j}{\bf i}'| p_l \\
		& = \frac{\theta_q(k,l,m)}{[k+1]_q}(Y(1,{\bf j}) + Y(2,{\bf j})).
	\end{align*} 
Thus, we have
	$$\Phi^{l,\bar{m}}_k(\rho) = \frac{1}{(N-1)^{l-r}}\sum_{{\bf j} \in I^{l-r}_{\ne}\backslash(1)}(Y(1,{\bf j}) + Y(2,{\bf j})) = Y(1) + Y(2),$$
which means
	$$Y(1) = \frac{[k+1]_q}{(N-1)^{l-r}\theta_q(k,l,m)}\sum_{{\bf j} \in I^{l-r}_{\ne}\backslash(1)}\sum_{{\bf i}\in (1,j_{l-r})/I^r_{\ne}} |{\bf j}{\bf i} \ra \la{\bf j}{\bf i}|.$$
Now we have
	\begin{align*}
	H(Y(1))
	& = (N-2)(N-1)^{r-1}\frac{[k+1]_q}{\theta_q(k,l,m)} \log \frac{(N-1)^{l-r}\theta_q(k,l,m)}{[k+1]_q}\\
	& = (1-\frac{2}{N}) (1-\frac{1}{N})^{r-1} \frac{N^r[k+1]_q}{\theta_q(k,l,m)} \log \frac{(N-1)^{l-r}\theta_q(k,l,m)}{[k+1]_q}\\
	& = (1+O(\frac{1}{N}))\log[ (1+O(\frac{1}{N}))N^l]
	\end{align*}
and
	$$\Tr(Y(2)) = 1 - \Tr(Y(1)) = 1 - (1-\frac{2}{N}) (1-\frac{1}{N})^{r-1} \frac{N^r[k+1]_q}{\theta_q(k,l,m)} = O(\frac{1}{N}).$$
Thus, we similarly get, by Lemma \ref{lem-Fannes-extend}, that $| H(\Phi^{l,\bar{m}}_k(\rho)) - H(Y(1)) | \lesssim O(\frac{1}{\sqrt{N}}).$

Combining all the above estimates we get
	$$\lim_{N\to \infty} |H(\Phi^{l,\bar{m}}_k(\rho)) - H(\Phi^{l,\bar{m}}_k(\rho)) - \frac{l+k-m}{2}\cdot \log N| = 0,$$
which gives us the desired lower estimate for $Q^{(1)}(\Phi^{l,\bar{m}}_k)$ as $N\to \infty$.

For the case $k=l+m$ we actually have the following exact formulae.
	$$\Phi^{l,\bar{m}}_{l+m}( \frac{1}{(N-1)^l}\sum_{{\bf j} \in I^l_{\ne}/(1)}| {\bf j}{\bf n} \ra \la {\bf j}{\bf n} |) = \frac{1}{(N-1)^l}\sum_{{\bf j} \in I^l_{\ne}/(1)}| {\bf j} \ra \la {\bf j} |$$
and
	$$\Phi^{\bar{l},m}_{l+m}( \frac{1}{(N-1)^l}\sum_{{\bf j} \in I^l_{\ne}/(1)}| {\bf j}{\bf n} \ra \la {\bf j}{\bf n} |) = | {\bf n} \ra \la {\bf n} |,$$
which tells us that	$Q^{(1)}(\Phi^{l,\bar m}_{l+m}) \ge l \cdot \log (N-1)$.

The estimates for $Q^{(1)}(\Phi^{\bar{l},m}_k)$ can be obtained in a similar way.

\end{proof}

Combining Theorem \ref{thm-MOE-sharp} and Theorem \ref{thm-QC}, we obtain the following asymptotically sharp one-shot capacities:

\begin{cor}\label{cor-capacities}
For each admissible triple $(k,l,m) \in \n^3_0$ we have 
$$\frac{l+k-m}{2}\log(N) -C_1(N)\leq Q^{(1)}(\Phi^{l, \bar m}_k) \leq \chi(\Phi^{l,\bar m}_k)\leq \frac{l+k-m}{2}\log(N) +C_2(N) $$
and
$$\frac{m+k-l}{2}\log(N) -D_1(N)\leq Q^{(1)}(\Phi^{l, \bar m}_k) \leq \chi(\Phi^{l,\bar m}_k)\leq \frac{m+k-l}{2}\log(N) +D_2(N) $$
with constants $C_1(N),C_2(N),D_1(N),D_2(N)\rightarrow 0$ as $N\rightarrow \infty$.
\end{cor}

\begin{proof}
Theorem \ref{thm-QC} directly gives us the wanted lower bounds, and Theorem \ref{thm-MOE-sharp} together with a general fact \eqref{eq-Holevo-MOE}
completes the conclusion.
\end{proof}

\begin{rem}
We note that Corollary \ref{cor-capacities} gives us asymptotically sharp ``one-shot'' private capacities $P^{(1)}(\Phi^{l, \bar m}_k)$ and $P^{(1)}(\Phi^{\bar l, m}_k)$ since
$$Q^{(1)}\leq P^{(1)}\leq \chi$$
in general. The one-shot private capacity $P^{(1)}$ is defined as
$$\max \left \{ H(\sum_x p_x\Phi(\rho_x))-\sum_x p_x H(\Phi(\rho_x)) -H(\sum_x p_x\widetilde{\Phi}(\rho_x))+\sum_x p_xH(\widetilde{\Phi}(\rho_x)) \right\}$$
where the maximum runs over all ensembles of quantum states $\left \{(p_x),(\rho_x)\right\}$. See \cite[Section 13.6]{Wi17} for details.
\end{rem}

\section{EBT/PPT and (anti-)degradability of TL-channels}\label{sec:EBP-PPT}

Since we have studied ``one-shot'' capacities $Q^{(1)}$ and $\chi$ for $O_N^+$-TL-channels in previous section, it is very natural to investigate their regularized quantities $Q$ and $C$. Since our $O_N^+$-TL-channels are bistochastic, we know that the classical capacity $C$ is smaller than $2\chi$ asymptotically by Proposition \ref{prop-bistochastic-estimates}:
$$C(\Phi^{l, \bar m}_k)\leq (l+k-m)\log(N),~C(\Phi^{\bar l, m}_k)\leq (m+k-l)\log(N) .$$

Although the regularized quantities $Q$ and $C$ are computationally intractible for many channels, some structural properties such as EBT/PPT/(anti-)degradability enable us to handle the regularization issues (See Proposition \ref{prop:implications}). However, we will show that our TL-channels associated with $O_N^+$ and $SU(2)$ have no such structural properties in most cases.

\subsection{The case of $O^+_N$}

\subsubsection{EBT property}
We now apply Theorem \ref{thm:choi-eq} to investigate EBT property for our $O_N^+$-TL-channels $\Phi_{k}^{\overline{l},m}$. Before coming to our result characterizing the EBT property for the channels $\Phi_k^{\overline{l},m}$, we first need an elementary lemma.

\begin{lem}\label{lem:ent-sub}
Let $H_A$ and $H_B$ be finite dimensional Hilbert spaces, let $0 \neq p \in B(H_B \otimes H_A)$ be an orthogonal projection, and let $H_0 \subseteq H_B \otimes H_A$ denote the range of $p$.  If $H_0$ is an entangled subspace of $H_B \otimes H_A$, then the state $\rho := \frac{1}{\mathrm{dim} H_0}p$ is entangled.
\end{lem}
\begin{proof}
We prove the contrapositive. If $\rho$ is separable, then we can write 
$$p = \sum_i |\xi_i \rangle \langle \xi_i| \otimes |\eta_i \rangle \langle \eta_i| \qquad (0 \neq  \xi_i \in H_B, \ 0 \ne \eta_i \in H_A).$$ 

For each $i$ put $x_i = |\xi_i \rangle \langle \xi_i| \otimes |\eta_i \rangle \langle \eta_i|$. Then since $x_i \leq p$ and $p$ is a projection, it follows that $x_i = px_ip$, which implies that the range of $x_i$ is contained in the range of $p$. In particular, $\xi_i \otimes \eta_i \in H_0$, so $H_0$ is separable.
\end{proof}

\begin{thm} \label{thm:EBT}
Let $(k,l,m) \in \N_0^3$ be an admissible triple. If $k \neq l-m$, then the quantum channel $\Phi_{k}^{\overline{l},m}$ is not EBT. Also, if $k\ne m-l$, then the quantum channel $\Phi_k^{l,\overline{m}}$ is not EBT.
\end{thm}

\begin{proof}
We have from Theorem \ref{thm:choi-eq} that $C_{\Phi_{k}^{\overline{l},m}} = \frac{[k+1]_q}{[l+1]_q} \alpha_l^{m,k}(\alpha_{l}^{m,k})^* \in \mc B(H_m \otimes H_k).$  Consider the orthogonal projection $p = \alpha_l^{m,k}(\alpha_{l}^{m,k})^*$.  The range of $p$ is the subrepresentation of $H_m \otimes H_k$ equivalent to $H_l$, and by [Theorem 3.2, \cite{BrCo18b}] this subspace is entangled iff $l\neq k+m$.   Applying Lemma  \ref{lem:ent-sub}, we conclude that $ \Phi_{k}^{\overline{l},m}$ is not EBT whenever $k \neq l-m$.
\end{proof}

\begin{rem}
We note that Theorem \ref{thm:EBT} leaves open whether or not the channels $\Phi_{l-m}^{\overline{l},m}$ are EBT.  In this case, the corresponding Choi map is a multiple of a projection onto a separable subspace, and we do not know if this projection is a multiple of an entangled state. 
\end{rem}

\subsubsection{PPT/ (anti-)degradability}

As the next step, one might naturally ask if $O_N^+$-TL-channels can have PPT property or (anti-)degradability. In fact, Theorem \ref{thm-QC} provides a strong partial answer on these structural questions for large $N$ as follows:

\begin{cor}
\begin{enumerate}
\item The channel $\Phi^{l, \bar m}_k$ is not PPT if $k>m-l$ and $\Phi^{\bar l, m}_k$ is not PPT if $k>l-m$ for sufficiently large $N$. In particular, the channels $\Phi^{l,\bar m}_{l+m}$ and $\Phi^{\bar l, m}_{l+m}$ are not PPT for all $N\geq 3$.
\item The channels $\Phi^{l,\bar m}_k$ and $\Phi^{\bar l, m}_k$ are neither degradable nor anti-degradable if $k>|l-m|$ for sufficiently large $N$.
\end{enumerate}
\end{cor}

\begin{proof}
\begin{enumerate}
\item Note that every PPT channel should have zero quantum capacity and that $Q(\Phi^{l, \bar m}_k)> 0$ if $k > m-l$ for sufficiently large $N$. Similar arguments are valid for $\Phi^{\bar l, m}_k$.
\item Note that every anti-degradable channel must have zero quantum capacity, while on the other hand both $\Phi^{l, \bar m}_k$ and $\Phi^{\bar l, m}_k$ have strictly positive quantum capacities for sufficiently large $N$ if $k>|l-m|$.
\end{enumerate}
\end{proof}

\subsection{The case of $SU(2)$}

We have a much better understanding about the TL-channels associated with $SU(2)$ than the ones from $O^+_N$ based on the following concrete description of Clebsch-Gordan coefficients. For an admissible triple $(k,l,m)\in \n^3_0$ we consider the associated isometry
\[\alpha^{l,m}_k |i\ra=\sum_{j=0}^l \sum_{j'=0}^m C^{l,m,k}_{j,j',i}|j j'\ra,\]

We actually have a precise but complicated formula (e.g. \cite[page 510]{VK}) for the constant $C^{l,m,k}_{j,j',i}$, which is a sum with multiple terms. Thus, the general constant $C^{l,m,k}_{j,j',i}$ is difficult to handle, but they satisfy several symmetries and some extremal cases can be written in a simpler form.
\begin{prop}\label{prop-symmetry} For any admissible triples $(k,l,m), (i,j,j') \in \n^3_0$ we have
	\begin{enumerate}
	
		\item $C^{l,m,k}_{j,j',i} = 0$ \, if\; $i + \frac{l+m-k}{2} \ne j+j'$,

		\item	$\begin{cases} \la i_1| \Phi^{l,\bar{m}}_k(|i\ra\la j |)|j_1\ra=0, & i_1-j_1\neq i-j\\ \la i_2| \Phi^{\bar{l},m}_k(|i\ra\la j |)|j_2\ra=0, & i_2-j_2\neq i-j\end{cases}$\; for \;$\begin{cases} 0\leq i_1,j_1\leq l,~0\leq i,j\leq k\\ 0\leq i_2,j_2\leq m,~0\leq i,j\leq k\end{cases},$

		\item $C^{l,m,k}_{j,j',i}=(-1)^{\frac{l+m-k}{2}}C^{m,l,k}_{j',j,i}$,
		
		\item $C^{l,m,k}_{j,j',i}=(-1)^{\frac{l+m-k}{2}}C^{l,m,k}_{l-j,m-j',k-i}$,
		
		\item $C^{l,m,k}_{j,j',i} \ne 0$\, if $\displaystyle i+\frac{l+m-k}{2}=j+j'$ and if one of the following is true: $\begin{cases} j=0, l \\ j'=0, m\\ i=0,k\end{cases}$.
		
	\end{enumerate}
\end{prop}
\begin{proof}
(2) We have $\la i_1| \Phi^{l,\bar{m}}_k(|i\ra\la j |)|j_1\ra=\displaystyle \sum_{i_2=0}^m C^{l,m,k}_{i_1,i_2,i}\overline{C^{l,m,k}_{j_1,i_2,j}}=0$ if $i_1-i\neq j_1-j$ by (1) and a similar argument holds for $\Phi^{\bar{l},m}_k$.

(5) If one of the parameters $i, j, j'$ becomes extremal, then the constant $C^{l,m,k}_{j,j',i}$ can be expressed in a single term, which is a ratio of several factorials by \cite[section 8.2.6]{VK} and the above symmetries (3) and (4).
\end{proof}

The $SU(2)$-TL-channel $\Phi^{l, \bar m}_k$ is of the following form.
	\begin{align}	
		\Phi^{l, \bar m}_k(|i\ra \la \tilde{i}|)
			& = (\iota \otimes \Tr)(\alpha^{l,m}_k |i\ra \la \tilde{i}| (\alpha^{l,m}_k)^*) \nonumber \\
			& = (\iota \otimes \Tr)(\sum_{j, \tilde{j}=0}^l \sum_{j', \tilde{j'}=0}^m C^{l,m,k}_{j,j',i}\overline{C^{l,m,k}_{\tilde{j},\tilde{j'},\tilde{i}}}|j j'\ra \la \tilde{j} \tilde{j'}|)\nonumber \\
			& = \sum_{j, \tilde{j}=0}^l \sum_{j'=0}^m C^{l,m,k}_{j,j',i}\overline{C^{l,m,k}_{\tilde{j},j',\tilde{i}}}\,|j\ra \la \tilde{j}|\\
			& = \sum_{j'=0}^m \sum_{j, \tilde{j}=0}^l C^{m,l,k}_{j',j,i}\overline{C^{m,l,k}_{j',\tilde{j},\tilde{i}}}\,|j\ra \la \tilde{j}| = \Phi^{\bar m, l}_k (|i\ra \la \tilde{i}|).\nonumber
	\end{align} 
The fourth equality is due to (3) of Proposition \ref{prop-symmetry}.

\begin{prop}
	For any admissible triple $(k,l,m) \in \n^3_0$ we have $\Phi^{l, \bar m}_k = \Phi^{\bar m, l}_k$. In particular, we have $\Phi^{l,\bar l}_k=\Phi^{\bar l, l}_k$, so that the channel $\Phi^{l,\bar l}_k$ is always degradable and anti-degradable.
\end{prop}

This allows us to restrict our attention to the case of $l\ge m$.

\subsubsection{EBT/PPT properties}

In this subsection, we completely characterize when the $SU(2)$-TL-channels $\Phi^{l,\overline{m}}_k$ and $\Phi^{\overline{l},m}_k$ are EBT or PPT. The main result of this subsection is as follows. 

\begin{thm}\label{thm:PPT1}
Let $(k,l,m)\in \n^3_0$ be an admissible triple with $l\geq m$.
\begin{enumerate}
\item The channel $\Phi^{l,\bar{m}}_k$ is EBT if and only if it is PPT if and only if $k=0$.
\item The channel $\Phi^{\bar{l},m}_k$ is EBT if and only if it is PPT if and only if $k=l-m$.
\end{enumerate}
\end{thm}
\begin{proof}
(1) If the channel $\Phi^{l,\bar{m}}_k$ is PPT, then its Choi matrix
\[C_{T\circ \Phi}=(T\circ \Phi\otimes \iota)(\sum_{i,j=1}^{d_A} |i\ra \la j|\otimes  |i\ra \la j|)=\sum_{i,j=1}^{d_A} T\circ \Phi(|i\ra \la j|)\otimes  |i\ra \la j|\]
should be a positive definite matrix. In particular, for any orthogonal unit vectors $v_1,v_2\in H_B\otimes H_A$ we should have
$$ \begin{bmatrix} \la v_1|C_{T\circ \Phi}|v_1\ra&\la v_1|C_{T\circ \Phi}|v_2\ra \\ 
\la v_2| C_{T\circ \Phi}|v_1\ra &\la v_2 | C_{T\circ \Phi}|v_2\ra   \end{bmatrix} = \begin{bmatrix} a &  b \\ 
\bar{b} & c \end{bmatrix} \ge 0.$$
We take a particular choice of $v_1$, $v_2$ as follows.
	$$\begin{cases} |v_1\ra = | l 0\ra, |v_2\ra = |0l\ra & \text{if}\; k>l\\ |v_1\ra = | l 0\ra, |v_2\ra = |l-k, k\ra & \text{if}\; k\le l\end{cases}.$$
Now we have $\displaystyle a =  \la v_1|C_{T\circ \Phi}|v_1\ra = \sum_{j'=0}^m C^{l,m,k}_{l,j',0}\overline{C^{l,m,k}_{l,j',0}}$. Since the channel $\Phi^{l,\bar{m}}_0$ is trivially EBT (and PPT) we may assume $k>0$, then $l + j' \ne \frac{l+m-k}{2}$ from the restriction that $l\ge m$. Thus, we get $a = 0$ by (1) of Proposition \ref{prop-symmetry}. Similarly, we can check that $b = C^{l,m,k}_{0, \frac{l+m-k}{2}, 0}\overline{C^{l,m,k}_{l, \frac{l+m-k}{2}, l}}$ for $k > l$. By (5) of Proposition \ref{prop-symmetry} we know that $b\ne 0$, so that ${\rm det}\begin{bmatrix} a &  b \\ 
\bar{b} & c \end{bmatrix} = - |b|^2 <0$, which is a contradition. The case $k \le l$ can be done by the same argument.

\vspace{0.3cm}
(2) We apply a similar argument as before. By taking
	$$\begin{cases} |v_1\ra = | m 0\ra, |v_2\ra = |m-k, k\ra & \text{if}\; l-m < k \le m\\ |v_1\ra = | m 0\ra, |v_2\ra = |0 m\ra & \text{if}\; k> m \lor l-m \end{cases}$$
we can similarly check that the matrix $ \begin{bmatrix} \la v_1|C_{T\circ \Phi}|v_1\ra&\la v_1|C_{T\circ \Phi}|v_2\ra \\ 
\la v_2| C_{T\circ \Phi}|v_1\ra &\la v_2 | C_{T\circ \Phi}|v_2\ra   \end{bmatrix}$ is not positive definite, so that the channels $\Phi^{\bar l, m}_k$ is not PPT if $k>l-m$.

But the case $k=l-m$ is no longer trivial. Note that we can pick a product vector $e\otimes f\in H_{l}\subseteq H_m\otimes H_{l-m}$ with $e\in H_m$ and $f\in H_{l-m}$. Then, by Theorem \ref{thm:choi-eq} and Proposition \ref{prop:ave}, we have
\begin{align*}
\frac{1}{l-m+1}C_{\Phi^{\overline{l},m}_{l-m}}&=\frac{1}{l+1}\alpha^{m,l-m}_l (\alpha^{m,l-m}_l)^*\\
&=\int_{SU(2)}\pi_m(x^{-1}) |e \ra\la e | \pi_m(x) \otimes \pi_{l-m}(x^{-1}) |f \ra\la f | \pi_{l-m}(x)dx,
\end{align*}
where $dx$ implies the normalized Haar measure on $SU(2)$.
This implies that the normalized Choi matrix of $\Phi^{\overline{l},m}_{l-m}$ is a separable state since the set of separable states are closed.
\end{proof}

\subsection{(Anti-)Degradability}\label{subsec:highest}

We first present the following cases when $SU(2)$-TL-channels are (anti-)degradable.

\begin{thm}\label{thm:highest}
Let $(k,l,m)\in \n^3_0$ be an admissible triple with $l\ge m$.

	\begin{enumerate}
		\item The channel $\Phi^{l,\bar{m}}_k$ is degradable if (a) $l=m$\, or\, (b) $k=l+m$ \, or \,(c) $k = l-m$. Moreover, we have a degrading channel for the highest weight case as follows.
	\begin{equation}\label{eq-degrading}
	\Phi^{m,\overline{l-m}}_l\circ \Phi^{l,\bar{m}}_{l+m}=\Phi^{m, \bar{l}}_{l+m}.
	\end{equation}
		\item The channel $\Phi^{l, \bar m}_k$ is not anti-degradable for $l>m$. Equivalently, $\Phi^{\bar l , m}_k$ is not degradable for $l>m$.
	\end{enumerate}
\end{thm}
\begin{proof}

(1) For the identity \eqref{eq-degrading} we need to show that for any $0\leq i,j\leq l+m$ and for any $s_2$ such that $\max \left \{0,i-j\right\}\leq s_2\leq \min \left \{m,m+i-j\right\}$,
\[(\Phi^{\overline{l-m},m}_l\circ \Phi^{l,\bar{m}}_{l+m}(|i\ra \la j |))_{s_2,s_2+j-i}=(\Phi^{\bar{l},m}_{l+m}(|i \ra \la j|))_{s_2,s_2+j-i}\]
by (2) of Proposition \ref{prop-symmetry}.

Equivalently, let us show that for any $\max \left \{0,i-j\right\}\leq s_2\leq \min \left \{m,m+i-j\right\}$
\[\sum_{i_2}C^{l,m,l+m}_{i-i_2,i_2,i}\overline{C^{l,m,l+m}_{j-i_2,i_2,j}}C^{l-m,m,l}_{i-i_2-s_2,s_2,i-i_2}\overline{C^{l-m,m,l}_{i-i_2-s_2,s_2+j-i,j-i_2}}=C^{l,m,l+m}_{i-s_2,s_2,i}\overline{C^{l,m,l+m}_{i-s_2,s_2+j-i,j}},\]
where $i_2$ runs over $\max \left \{0,i-s_2-l+m\right\}\leq i_2\leq \min \left \{m,i-s_2\right\}$.

We use the following explicit formula for Clebsch-Gordan coefficients to the highest weight case, namely for any $l,m$
\[C^{l,m,l+m}_{j_1,j_2,j}=\delta_{j_1+j_2,j}\sqrt{\frac{l! m!}{(l+m)!}}\sqrt{\frac{j!(l+m-j)!}{j_1! j_2! (l-j_1)!(m-j_2)!}}.\]

Now, we have
\begin{align*}
\lefteqn{\sum_{i_2}C^{l,m,l+m}_{i-i_2,i_2,i}\overline{C^{l,m,l+m}_{j-i_2,i_2,j}}C^{l-m,m,l}_{i-i_2-s_2,s_2,i-i_2}\overline{C^{l-m,m,l}_{i-i_2-s_2,s_2+j-i,j-i_2}}}\\
\notag&=\frac{l!m!}{(l+m)!}\frac{(l-m)!m!}{l!}\sqrt{\frac{i!(l+m-i)!j!(l+m-j)!}{s_2!(m-s_2)!(s_2+j-i)!(m-s_2-j+i)!}}\\
\notag&\;\;\;\; \times \sum_{i_2} \frac{1}{i_2!(m-i_2)!(i-i_2-s_2)!(l-m+s_2+i_2-i)!}\\
\notag&= \frac{l!m!}{(l+m)!l!}\sqrt{\frac{i!(l+m-i)!j!(l+m-j)!}{s_2!(m-s_2)!(s_2+j-i)!(m-s_2-j+i)!}}\sum_{i_2} {m \choose i_2}{l-m \choose i-s_2-i_2}\\
\label{eq1}&=\frac{l!m!}{(l+m)!l!}\sqrt{\frac{i!(l+m-i)!j!(l+m-j)!}{s_2!(m-s_2)!(s_2+j-i)!(m-s_2-j+i)!}}{l \choose i-s_2}\\
\notag&= \frac{l!m!}{(l+m)!}\frac{1}{(i-s_2)!(l+s_2-i)!}\sqrt{\frac{i!(l+m-i)!j!(l+m-j)!}{s_2!(m-s_2)!(s_2+j-i)!(m-s_2-j+i)!}}\\
\notag&= \sqrt{\frac{l!m!}{(l+m)!}}\sqrt{\frac{i!(l+m-i)!}{(i-s_2)!s_2!(l-i+s_2)!(m-s_2)!}}\\
\notag& \;\;\;\; \times \sqrt{\frac{l!m!}{(l+m)!}}\sqrt{\frac{j!(l+m-j)!}{(i-s_2)!(s_2+j-i)!(l-i+s_2)!(m-s_2-j+i)!}}\\
\notag&= C^{l,m,l+m}_{i-s_2,s_2,i}\overline{C^{l,m,l+m}_{i-s_2,s_2+j-i,j}}.
\end{align*}

The third equality in the above is from the following fact
\[{l \choose i-s_2}=\sum_{\max \left \{0,i-s_2-l+m \right\}\leq i_2\leq \min \left \{m,i-s_2\right\}} {m \choose i_2} {l-m \choose i-s_2-i_2}.\]

\vspace{0.3cm}

(2) By Proposition \ref{prop-bistochastic-estimates} and Proposition \ref{prop-CGchannel-bistochastic} we know that
\[0< \log(\frac{l+1}{m+1}) \leq Q^{(1)}(\Phi^{l, \bar m}_k),\]
which leads us to the conlusion we wanted.
\end{proof}

\begin{ex}\label{ex:non-deg-non-antideg}

The channel $\Phi^{l, \bar m}_k$ could be non-degradable for intermediate $l-m<k<l+m$ at least in low dimensional examples. The strategy is to find explicit states $\rho\in M_{k+1}$ such that 
\[0<H(\Phi^{\bar l, m}(\rho))-H(\Phi^{l, \bar m}_k(\rho))\leq Q^{(1)}(\Phi^{\bar l, m}_k).\]
The inequality above implies that $\Phi^{\bar l, m}_k$ is not anti-degradable and equivalently $\Phi^{l, \bar m}_k$ is not degradable.

The channel $\Phi^{3,\bar 2}_3$ is not degradable. Indeed, for $\rho=\displaystyle \left [ \begin{array}{cccc} 0.25&0&0&0\\ 0&0.75&0&0 \\ 0&0&0&0\\0&0&0&0 \end{array} \right ]$ we have
\begin{align*}
&H(\Phi^{\bar 3, 2}_3(\rho))-H(\Phi^{3, \bar 2}_3(\rho))\\
& = H(\left [ \begin{array}{ccc}
0.5&0&0\\ 0&0.2&0\\0&0&0.3
\end{array} \right ]
)-H(\left [ \begin{array}{cccc} 
0.45 &0&0&0\\ 0&0.15&0&0\\
0&0&0.4&0\\
0&0&0&0
\end{array} \right ])\\
&\approx 0.0192,
\end{align*}
where the first equality is obtained by the precise description of the associated isometry
	$$\alpha^{3,2}_3:\displaystyle \Comp^4\rightarrow \Comp^4\otimes \Comp^3, \begin{array}{ll} |1\ra &\mapsto -\sqrt{\frac{3}{5}} |12\ra + \sqrt{\frac{2}{5}} |21\ra \\ |2\ra & \mapsto -\sqrt{\frac{2}{5}} |13\ra -\sqrt{\frac{1}{15}} |22\ra+\sqrt{\frac{8}{15}}|31\ra \end{array}$$
using the known formula of Clebsch-Gordan coefficients. Here, $\left \{|j\ra\right\}_{j=1}^{n+1}$ refers to the canonical orthonormal basis of $H_n=\Comp^{n+1}$.
\end{ex}

\begin{rem}
For the channels $\Phi^{\bar{l},m}_{l+m}$ with $l\geq m$, we have 
\[0=Q(\Phi^{\bar{l},m}_{l+m})<C(\Phi^{\bar{l},m}_{l+m})=\log(\mathrm{dim}(H_m))\] 
by Theorem \ref{thm:highest}, \cite[Proposition 8.8]{Hol-book} and the fact that $H_{\min}(\Phi^{l, \bar{m}}_{l+m})=0$. This means that, outside the realm of entanglement-breaking channels, there exist channels which completely destroy quantum information though all the classical information can be preserved.
\end{rem}

\section{Tensor products of Temeperley-Lieb Channels and outputs of Entangled Covariant States}\label{sec:tensor}

It is well known that additivity of Holevo capacities is equivalent to additivity of minimum output entropies \cite{Sh04b} and Hastings \cite{Ha09} established non-additivity of the minimum output entropy by exhibiting the existence of random unitary channels $\Phi$ such that 
\begin{equation}\label{ineq-Hastings}
H_{\min}(\Phi\otimes \overline{\Phi})< H_{\min}(\Phi)+H_{\min}(\overline{\Phi}),
\end{equation}
where $\overline{\Phi}$ is the conjugate channel of $\Phi$. In the proof of (\ref{ineq-Hastings}), the maximally entangled state was used to estimate an upper bound of $H_{\min}(\Phi\otimes \overline{\Phi})$. Since we know the minimum output entropies for single $O_N^+$-TL-channels in an asymptotic sense, it is natural to try to evaluate the minimum output entropies for tensor products of $O_N^+$-TL-channels.  Although we are unable to fully evaluate such minimum output entropies for all tensor products, we do establish upper bounds for the minimum output entropies $H_{\min}(\Phi^{\bar l_1, m_1}_{k_1}\otimes \Phi^{l_2, \bar m_2}_{k_2})$.  This is achieved by evaluating the entropies $H((\Phi^{\bar l_1, m_1}_{k_1}\otimes \Phi^{l_2, \bar m_2}_{k_2})(\rho))$ for certain entangled states $\rho$. More precisely, we will present explicit formulae for
$$H((\Phi^{\bar l_1, m_1}_{k_1}\otimes \Phi^{l_2, \bar m_2}_{k_2})(\frac{1}{[i+1]_q}\alpha^{k_1,k_2}_i (\alpha^{k_1,k_2}_i)^*))$$
for all admissible triples $(i,k_1,k_2)\in \n^3_0$.

In this section we use all the notation and planar string diagram formalism for $\text{Rep}(O^+_F)$ introduced in Section \ref{TL-diagrams}.

\subsection{Tetrahedral nets and the quantum $6j$-symbols}

Following \cite{KaLi94}, let $\mc A \subset \N_0^6$ be the set of all sextuples $\left [\begin{matrix} a & b &i \\ c& d&j
\end{matrix}\right]$ with the property that each of the following triples \[(a,d,i), \ (b,c,i), \ (a,b ,j), \ (d,c,j)\] is admissible.   We define the \textit{tetrahedral net} to be the  function $\text{Tet}_q:\mc A \to \C$ given by
\[
\text{Tet}_q\left [\begin{matrix} a & b &i \\ c& d&j
\end{matrix}\right]  = \tau_i((A_i^{b,c})^* (\iota_{H_b} \otimes (A_0^{j,j})^* \otimes \iota_{H_c})(A_a^{b,j} \otimes A_d^{j,c})A_i^{a,d}).
\] 
In terms of planar string diagrams, the Tet$_q$ functions are given by
\[\text{Tet}_q\left [\begin{matrix} a & b &i \\ c& d&j
\end{matrix}\right] = \begin{tikzpicture}[baseline=(current  bounding  box.center),
			wh/.style={circle,draw=black,thick,inner sep=.5mm},
			bl/.style={circle,draw=black,fill=black,thick,inner sep=.5mm}, scale = 0.2]

\node  at (-1,-5) {$i$};	
\node at (-1,13) {$i$};
\node at (-2,1) {$a$};
\node at (2,1) {$d$};
\node  at (0,5) {$j$};	
\node at (-2,7) {$b$};
\node at (2,7) {$c$};

\draw [-, color=black]
		(-4,4) -- (4, 4);

\draw [-, color=black]
		(-4,4) -- (0,0);

\draw [-, color=black]
		(0,-0) -- (4,4);

\draw [-, color=black]
		(0,-4) -- (0,0);
		
\draw [-, color=black]
		(-4,4) -- (0,8);
		
\draw [-, color=black]
		(4,4) -- (0,8);
		
\draw [-, color=black]
		(0,8) -- (0,12);		

\draw [-, color=black]
		(0,12) to [bend left =110] (0,-4);

	\end{tikzpicture} . \]

Next, we introduce the {\it quantum $6j$-symbols} $\{\cdot\}_q: \mc A \to \C$, which are defined  in terms of the tetrahedral nets as follows:

\[
\left\{\begin{matrix} a & b &i \\ c& d&j
\end{matrix}\right\}_q = \frac{\text{Tet}_q\left [\begin{matrix} a & b &i \\ c& d&j
\end{matrix}\right][i+1]_q}{\theta_q(a,d,i)\theta_q(b,c,i)}.\]

\begin{rem}
We note that there exist simple algebraic formulae that allow one to numerically evaluate the tetrahedral nets (and hence also the quantum $6j$-symbols).  See \cite[Section 9.11]{KaLi94} for example.
\end{rem}

The most important geometric-algebraic feature of the quantum $6j$-symbols $\left\{\begin{matrix} a & b &i \\ c& d&j
\end{matrix}\right\}_q$ is that they arise as the basis change coefficients for two canonical bases for the Hom-space $\text{Hom}_{O^+_F}(H_a \otimes H_d, H_b  \otimes H_c)$.  More precisely, $\text{Hom}_{O^+_F}(H_a \otimes H_d,  H_b \otimes H_c)$ has one linear basis given by the string diagrams 
\[
\begin{tikzpicture}[baseline=(current  bounding  box.center),
			wh/.style={circle,draw=black,thick,inner sep=.5mm},
			bl/.style={circle,draw=black,fill=black,thick,inner sep=.5mm}, scale = 0.2]

\node at (1,0) {$i$};
\node at (3,3) {$c$};
\node at (-3,3) {$b$};
\node at (3,-3) {$d$};
\node at (-3,-3) {$a$};

\draw [-, color=black]
		(0,-2) -- (0, 2);

\draw [-, color=black]
		(0,2) -- (2, 4);

\draw [-, color=black]
		(0,2) -- (-2, 4);

\draw [-, color=black]
		(0,-2) -- (-2, -4);

\draw [-, color=black]
		(0,-2) -- (2, -4);

	\end{tikzpicture} \qquad (i \in \N_0 \text{ such that }(i,a,d), \ (i,b,c) \text{ admissible}), 
\]   
and another linear basis given by
\[
\begin{tikzpicture}[baseline=(current  bounding  box.center),
			wh/.style={circle,draw=black,thick,inner sep=.5mm},
			bl/.style={circle,draw=black,fill=black,thick,inner sep=.5mm}, scale = 0.2]

\node at (0,0) {$j$};
\node at (3,3) {$c$};
\node at (-3,3) {$b$};
\node at (3,-3) {$d$};
\node at (-3,-3) {$a$};

\draw [-, color=black]
		(2,-2) -- (2, 4);

\draw [-, color=black]
		(-2,-2) -- (-2, 4);

\draw [-, color=black]
		(-2,1) -- (2,1);

	\end{tikzpicture} \qquad (j \in \N_0 \text{ such that }(j,a,b), \ (j,c,d) \text{ admissible}).  
\]   
We then have that the quantum $6j$-symbols are the basis change coefficients between these two bases \cite[Proposition 11] {KaLi94}: 
\begin{eqnarray}\label{6j1}\begin{tikzpicture}[baseline=(current  bounding  box.center),
			wh/.style={circle,draw=black,thick,inner sep=.5mm},
			bl/.style={circle,draw=black,fill=black,thick,inner sep=.5mm}, scale = 0.2]

\node at (0,0) {$j$};
\node at (3,3) {$c$};
\node at (-3,3) {$b$};
\node at (3,-3) {$d$};
\node at (-3,-3) {$a$};

\draw [-, color=black]
		(2,-2) -- (2, 4);

\draw [-, color=black]
		(-2,-2) -- (-2, 4);

\draw [-, color=black]
		(-2,1) -- (2,1);

	\end{tikzpicture}= \sum_{i}\left\{\begin{matrix} a & b &i \\ c& d&j
\end{matrix}\right\}_q  \begin{tikzpicture}[baseline=(current  bounding  box.center),
			wh/.style={circle,draw=black,thick,inner sep=.5mm},
			bl/.style={circle,draw=black,fill=black,thick,inner sep=.5mm}, scale = 0.2]

\node at (1,0) {$i$};
\node at (3,3) {$c$};
\node at (-3,3) {$b$};
\node at (3,-3) {$d$};
\node at (-3,-3) {$a$};

\draw [-, color=black]
		(0,-2) -- (0, 2);

\draw [-, color=black]
		(0,2) -- (2, 4);

\draw [-, color=black]
		(0,2) -- (-2, 4);

\draw [-, color=black]
		(0,-2) -- (-2, -4);

\draw [-, color=black]
		(0,-2) -- (2, -4);

	\end{tikzpicture}, 
\end{eqnarray}
and similarly by a rotational symmetry argument, 
\begin{eqnarray} \label{6j2}\begin{tikzpicture}[baseline=(current  bounding  box.center),
			wh/.style={circle,draw=black,thick,inner sep=.5mm},
			bl/.style={circle,draw=black,fill=black,thick,inner sep=.5mm}, scale = 0.2]

\node at (1,0) {$j$};
\node at (3,3) {$b$};
\node at (-3,3) {$a$};
\node at (3,-3) {$c$};
\node at (-3,-3) {$d$};

\draw [-, color=black]
		(0,-2) -- (0, 2);

\draw [-, color=black]
		(0,2) -- (2, 4);

\draw [-, color=black]
		(0,2) -- (-2, 4);

\draw [-, color=black]
		(0,-2) -- (-2, -4);

\draw [-, color=black]
		(0,-2) -- (2, -4);

	\end{tikzpicture}= \sum_{i}\left\{\begin{matrix} a & b &i \\ c& d&j
\end{matrix}\right\}_q\begin{tikzpicture}[baseline=(current  bounding  box.center),
			wh/.style={circle,draw=black,thick,inner sep=.5mm},
			bl/.style={circle,draw=black,fill=black,thick,inner sep=.5mm}, scale = 0.2]

\node at (0,0) {$i$};
\node at (3,3) {$c$};
\node at (-3,3) {$b$};
\node at (3,-3) {$d$};
\node at (-3,-3) {$a$};

\draw [-, color=black]
		(2,-2) -- (2, 4);

\draw [-, color=black]
		(-2,-2) -- (-2, 4);

\draw [-, color=black]
		(-2,1) -- (2,1);

	\end{tikzpicture}.
\end{eqnarray}

The following formula involving three-vertices and tetrahedral nets will be handy in the next subsection.

\begin{lem} \label{tetra-lemma}
Let  $\left [\begin{matrix} a & b &i \\ c& d&j
\end{matrix}\right] \in \mc A$.  Then 
\[
\begin{tikzpicture}[baseline=(current  bounding  box.center),
			wh/.style={circle,draw=black,thick,inner sep=.5mm},
			bl/.style={circle,draw=black,fill=black,thick,inner sep=.5mm}, scale = 0.2]

\node  at (-1,-5) {$i$};	
\node at (-2,1) {$a$};
\node at (2,1) {$d$};
\node  at (0,5) {$j$};	
\node at (-6,7) {$b$};
\node at (6,7) {$c$};

\draw [-, color=black]
		(-4,4) -- (4, 4);

\draw [-, color=black]
		(-4,4) -- (0,0);

\draw [-, color=black]
		(0,-0) -- (4,4);

\draw [-, color=black]
		(0,-4) -- (0,0);
		
\draw [-, color=black]
		(-4,4) -- (-8,8);
		
\draw [-, color=black]
		(4,4) -- (8,8);

	\end{tikzpicture}  =\frac{ \text{Tet}_q\left [\begin{matrix} a & b &i \\ c& d&j
\end{matrix}\right]}{\theta_q(i,b,c)}\begin{tikzpicture}[baseline=(current  bounding  box.center),
			wh/.style={circle,draw=black,thick,inner sep=.5mm},
			bl/.style={circle,draw=black,fill=black,thick,inner sep=.5mm}, scale = 0.2]

\node  at (-1,-5) {$i$};	

\node at (-6,7) {$b$};
\node at (6,7) {$c$};

\draw [-, color=black]
		(-4,4) -- (0,0);

\draw [-, color=black]
		(0,-0) -- (4,4);

\draw [-, color=black]
		(0,-4) -- (0,0);
		
\draw [-, color=black]
		(-4,4) -- (-8,8);
		
\draw [-, color=black]
		(4,4) -- (8,8);

	\end{tikzpicture}   .
\]
\end{lem}

\begin{proof}
Denote the quantity on the left hand side by $B$.  Then $B \in \text{Hom}_{O^+_F}(H_i, H_b\otimes H_c) = \C A_{i}^{b,c}$, and so there exists $\lambda  \in \C$ such that $B = \lambda A_{i}^{b,c}$ (i.e., $B$ is a multiple of a three-vertex).  But then we have \[\text{Tet}_q\left [\begin{matrix} a & b &i \\ c& d&j
\end{matrix}\right] = \tau_i((A_{i}^{b,c})^*B) = \tau_i((A_{i}^{b,c})^*\lambda A_{i}^{b,c}) = \lambda \theta_q(i,b,c). \]
\end{proof}

\subsection{Tensor products of TL-channels and outputs of entangled states}
Here we address tensor products of the form $\Phi_{k_1}^{\bar l_1, m_1} \otimes \Phi_{k_2}^{l_2, \bar m_2}$, and compute explicitly the outputs of $O^+_F$-covariant states of the form $\rho_{i}^{k_1, k_2} = \frac{1}{[i+1]}_q \alpha_i^{k_1, k_2}(\alpha_i^{k_1, k_2})^*$, for all admissible triples $(i,k_1,k_2)$.  Note that in the special case of $i = 0$ and $k_1 = k_2$, we have that $\rho_0^{k,k}$ is a maximally entangled state, and in general, $\rho_i^{k_1, k_2}$ is an entangled state \cite[Theorem 5.5]{BrCo17b} if $k_1,k_2>0$.

In order to ease the notational burden on the following theorem, let us fix once and for all admissible triples $(i, k_1, k_2)$, $(k_j, l_j, m_j) \in \N_0^3$ ($j = 1,2$), and let $X_i  = \big( \Phi_{k_1}^{\bar l_1, m_1} \otimes \Phi_{k_2}^{l_2, \bar m_2}\big)(\rho_i^{k_1, k_2})$

\begin{thm}
We have the following spectral decomposition for $X_i$:
\[
X_i =\sum_{\substack{l=m_1+l_2-2r \\
0 \le r \le \min\{m_1,l_2\}}} \lambda_{i, l}^{m_1, l_2} \alpha_l^{m_1,l_2}(\alpha_l^{m_1,l_2})^*, 
\]
where 
\begin{align*}
\lambda_{i, l}^{m_1, l_2}
&=\Bigg(\frac{[i+1]_q[k_1+1]_q[k_2+1]_q\theta_q(l, m_1,l_2)}{[l+1]_q\theta_q(k_1,l_1,m_1)\theta_q(k_2,l_2,m_2)\theta_q(i,k_1,k_2)}\Bigg) \\
&\;\;\;\;\times \sum_{\substack{j=2t\\ 0 \le t \le \min\{k_1,k_2 \}}}
\frac{ \text{\small $\left\{\begin{matrix} k_1 & k_2 & j \\ k_2& k_1& i
 \end{matrix}\right\}_q 
\text{Tet}_q\left [\begin{matrix} l_1 & m_1 &m_1 \\ j& k_1&k_1
\end{matrix}\right]\text{Tet}_q\left [\begin{matrix} k_2 & j &l_2 \\ l_2& m_2&k_2
\end{matrix}\right]
\left\{\begin{matrix} m_1 & m_1 &l \\ l_2& l_2&j
\end{matrix}\right\}_q$}}{\theta_q (m_1,m_1, j)\theta_q (l_2,j,l_2)}, 
\end{align*}
and occurs with multiplicity $[l+1]_q$.
\end{thm} 

\begin{proof}
We have that, up to planar isotopy, the planar tangle representating $X_i$ is given by:
\[
X_i = \frac{[i+1]_q[k_1+1]_q[k_2+1]_q}{\theta_q (k_1,k_2,i)\theta_q (l_1,m_1,k_1) \theta_q (l_2,m_2,k_2)} \begin{tikzpicture}[baseline=(current  bounding  box.center),
			wh/.style={circle,draw=black,thick,inner sep=.5mm},
			bl/.style={circle,draw=black,fill=black,thick,inner sep=.5mm}, scale = 0.17]
\node at (1,0) {$i$};
\node at (3,3) {{\footnotesize $ k_2$}};
\node at (-3,3) {{\footnotesize $ k_1$}};
\node at (-5,8) {$m_1$};
\node at (5,8) {$l_2$};
\node at (-5,-8) {$m_1$};
\node at (5,-8) {$l_2$};
\node at (3,-3) {{\footnotesize $k_2$}};
\node at (-3,-3) {{\footnotesize $k_1$}};
\node at (-7,0) {$l_1$};
\node at (8,0) {$m_2$};
\draw [-, color=black]
		(0,-2) -- (0, 2);
\draw [-, color=black]
		(0,2) -- (3, 5);
\draw [-, color=black]
		(0,2) -- (-3, 5);
\draw [-, color=black]
		(0,-2) -- (-3, -5);
\draw [-, color=black]
		(0,-2) -- (3, -5);
\draw [-, color=black]
		(-3,-9) -- (-3, -5);
\draw [-, color=black]
		(3,-9) -- (3, -5);
\draw [-, color=black]
		(-3,5) -- (-3,9);
\draw [-, color=black]
		(3,5) -- (3, 9);
\draw [-, color=black]
		(-3,5) to [bend right = 90]  (-3,-5);
\draw [-, color=black]
		(3,5) to [bend left = 90]  (3, -5);
	\end{tikzpicture}
\]

Using the formulae \eqref{6j1}-\eqref{6j2} for the quantum $6j$-symbols together with Lemma \ref{tetra-lemma}, we have 
\begin{align*}
&\begin{tikzpicture}[baseline=(current  bounding  box.center),
			wh/.style={circle,draw=black,thick,inner sep=.5mm},
			bl/.style={circle,draw=black,fill=black,thick,inner sep=.5mm}, scale = 0.17]
\node at (1,0) {$i$};
\node at (3,3) {{\footnotesize $ k_2$}};
\node at (-3,3) {{\footnotesize $ k_1$}};
\node at (-5,8) {$m_1$};
\node at (5,8) {$l_2$};
\node at (-5,-8) {$m_1$};
\node at (5,-8) {$l_2$};
\node at (3,-3) {{\footnotesize $k_2$}};
\node at (-3,-3) {{\footnotesize $k_1$}};
\node at (-7,0) {$l_1$};
\node at (8,0) {$m_2$};
\draw [-, color=black]
		(0,-2) -- (0, 2);
\draw [-, color=black]
		(0,2) -- (3, 5);
\draw [-, color=black]
		(0,2) -- (-3, 5);
\draw [-, color=black]
		(0,-2) -- (-3, -5);
\draw [-, color=black]
		(0,-2) -- (3, -5);
\draw [-, color=black]
		(-3,-9) -- (-3, -5);
\draw [-, color=black]
		(3,-9) -- (3, -5);
\draw [-, color=black]
		(-3,5) -- (-3,9);
\draw [-, color=black]
		(3,5) -- (3, 9);
\draw [-, color=black]
		(-3,5) to [bend right = 90]  (-3,-5);
\draw [-, color=black]
		(3,5) to [bend left = 90]  (3, -5);
	\end{tikzpicture}
= \sum_j
\left\{\begin{matrix} k_1 & k_2 & j \\ k_2& k_1& i
 \end{matrix}\right\}_q \begin{tikzpicture}[baseline=(current  bounding  box.center),
			wh/.style={circle,draw=black,thick,inner sep=.5mm},
			bl/.style={circle,draw=black,fill=black,thick,inner sep=.5mm}, scale = 0.17]
\node at (0,-1.2) {\footnotesize $j$};
\node at (1.9,3) {\footnotesize $k_2$};
\node at (-1.9,3) {\footnotesize $k_1$};
\node at (-5,8) {$m_1$};
\node at (5,8) {$l_2$};
\node at (-5,-8) {$m_1$};
\node at (5,-8) {$l_2$};
\node at (1.9,-3) {\footnotesize$k_2$};
\node at (-1.9,-3) {\footnotesize $k_1$};
\node at (-7,0) {\footnotesize$l_1$};
\node at (7.5,0) {\footnotesize $m_2$};
\draw [-, color=black]
		(-3,0) -- (3, 0);
\draw [-, color=black]
		(-3,-5) -- (-3, 5);
\draw [-, color=black]
		(3,-5) -- (3, 5);
\draw [-, color=black]
		(-3,-9) -- (-3, -5);
\draw [-, color=black]
		(3,-9) -- (3, -5);
\draw [-, color=black]
		(-3,5) -- (-3,9);
\draw [-, color=black]
		(3,5) -- (3, 9);
\draw [-, color=black]
		(-3,5) to [bend right = 90]  (-3,-5);
\draw [-, color=black]
		(3,5) to [bend left = 90]  (3, -5);
	\end{tikzpicture} 
= \sum_j 
\left\{\begin{matrix} k_1 & k_2 & j \\ k_2& k_1& i
 \end{matrix}\right\}_q \begin{tikzpicture}[baseline=(current  bounding  box.center),
			wh/.style={circle,draw=black,thick,inner sep=.5mm},
			bl/.style={circle,draw=black,fill=black,thick,inner sep=.5mm}, scale = 0.17]
\draw [-, color=black]
		(-1,0) to [bend right = -90] (1, 0);
\node at (0,-1) {\footnotesize$j$};
\draw [-, color=black]
		(-3,-5) -- (-1, 0);
\node at (-1,-3) {\footnotesize$k_1$};
\draw [-, color=black]
		(-5,0) -- (-1, 0);
\node at (-2.5,1) {\footnotesize$k_1$};
\draw [-, color=black]
		(3,-5) -- (1, 0);
\node at (1,-3) {\footnotesize$k_2$};
\draw [-, color=black]
		(1,0) -- (5, 0);
\node at (2.5,1) {\footnotesize$k_2$};
\draw [-, color=black]
		(-3,-9) -- (-3, -5);
\node at (-5,-8) {$m_1$};
\draw [-, color=black]
		(3,-9) -- (3, -5);
\node at (5,-8) {$l_2$};
\draw [-, color=black]
		(-5,0) -- (-5,4);
\node at (-5,5) {$m_1$};
\draw [-, color=black]
		(5,0) -- (5, 4);
\node at (5,5) {$l_2$};
\draw [-, color=black]
		(-3,-5) -- (-5,0);
\node at (-5,-2) {\footnotesize$l_1$};
\draw [-, color=black]
		(3,-5) -- (5, 0);
\node at (6,-2) {\footnotesize$m_2$};
	\end{tikzpicture} \\
&=\sum_j 
\left\{\begin{matrix} k_1 & k_2 & j \\ k_2& k_1& i
 \end{matrix}\right\}_q \frac{\text{ \small $ \text{Tet}_q\left [\begin{matrix} l_1 & m_1 &m_1 \\ j& k_1&k_1
\end{matrix}\right]\text{Tet}_q\left [\begin{matrix} k_2 & j &l_2 \\ l_2& m_2&k_2
\end{matrix}\right]$}}{\theta_q(m_1, m_1, j)\theta_q(l_2, j, l_2)} \begin{tikzpicture}[baseline=(current  bounding  box.center),
			wh/.style={circle,draw=black,thick,inner sep=.5mm},
			bl/.style={circle,draw=black,fill=black,thick,inner sep=.5mm}, scale = 0.17]
\node at (0,-0.1) {\footnotesize$j$};
\node at (3.5,3) {\footnotesize$l_2$};
\node at (-3.5,3) {\footnotesize$m_1$};
\node at (3.5,-3) {\footnotesize$l_2$};
\node at (-3.5,-3) {\footnotesize$m_1$};
\draw [-, color=black]
		(2,-2) -- (2, 4);
\draw [-, color=black]
		(-2,-2) -- (-2, 4);
\draw [-, color=black]
		(-2,1) -- (2,1);
	\end{tikzpicture} \\
&= \sum_{l}\sum_j \left\{\begin{matrix} k_1 & k_2 & j \\ k_2& k_1& i
 \end{matrix}\right\}_q \frac{ \text{ \small $ \text{Tet}_q\left [\begin{matrix} l_1 & m_1 &m_1 \\ j& k_1&k_1
\end{matrix}\right]\text{Tet}_q\left [\begin{matrix} k_2 & j &l_2 \\ l_2& m_2&k_2
\end{matrix}\right]$}}{\theta_q(m_1, m_1, j)\theta_q(l_2, j, l_2)} \left\{\begin{matrix} m_1 & m_1 &l \\ l_2& l_2&j
\end{matrix}\right\}_q  \begin{tikzpicture}[baseline=(current  bounding  box.center),
			wh/.style={circle,draw=black,thick,inner sep=.5mm},
			bl/.style={circle,draw=black,fill=black,thick,inner sep=.5mm}, scale = 0.17]
\node at (1,0) {\footnotesize$l$};
\node at (3,3) {\footnotesize$l_2$};
\node at (-3,3) {\footnotesize$m_1$};
\node at (3,-3) {\footnotesize$l_2$};
\node at (-3,-3) {\footnotesize$m_1$};
\draw [-, color=black]
		(0,-2) -- (0, 2);
\draw [-, color=black]
		(0,2) -- (2, 4);
\draw [-, color=black]
		(0,2) -- (-2, 4);
\draw [-, color=black]
		(0,-2) -- (-2, -4);
\draw [-, color=black]
		(0,-2) -- (2, -4);
	\end{tikzpicture} \\
&=\sum_l \Big( \sum_j \text{ \small $\left\{\begin{matrix} k_1 & k_2 & j \\ k_2& k_1& i
 \end{matrix}\right\}_q $} \frac{\text{ \small $ \text{Tet}_q\left [\begin{matrix} l_1 & m_1 &m_1 \\ j& k_1&k_1
\end{matrix}\right]\text{Tet}_q\left [\begin{matrix} k_2 & j &l_2 \\ l_2& m_2&k_2
\end{matrix}\right]$}}{\theta_q(m_1, m_1, j)\theta_q(l_2, j, l_2)} \text{ \small $ \left\{\begin{matrix} m_1 & m_1 &l \\ l_2& l_2&j
\end{matrix}\right\}_q $}\Big) \frac{\theta_q(l,m_1, l_2)}{[l+1]_q}  \alpha_l^{m_1,l_2}\alpha_l^{m_1,l_2*}.
\end{align*}

In the above, the summands run over $l$ such that $(l, m_1, l_2)$ is admissible, and $j$ such that both $(j, k_1, k_1)$ and $(j, k_2, k_2)$ are admissible.  This corresponds exactly to $l=m_1+l_2-2r$ with $0 \le r \le \min\{m_1,l_2\}$ and $j=2t$ with $0 \le t \le \min\{k_1,k_2 \}$.  The claimed formula for the eigenvalue $\lambda_{i, l}^{m_1, l_2}$ is now immediate.  Note also that the multiplicity of $\lambda_{i, l}^{m_1, l_2}$ is $\text{rank}(\alpha_l^{m_1,l_2} (\alpha_l^{m_1,l_2})^* )= \dim H_l = [l+1]_q$.
\end{proof}

\begin{rem}
As remarked above, the element $X_0 \in B(H_{m_1} \otimes H_{l_2})$ is the output of the $O^+_F$-covariant Bell state $\rho_0^{k,k} \in B(H_k \otimes H_k)$.  In this situation, the eigenvalue formula for $X_0$ simplifies greatly.  This can be seen by using similar arguments to those in the proof given above, or by directly using algebraic relations satisfied by the quantum $6j$-symbols.  In any case, we get    
\[X_0 = \sum_{\substack{l=m_1+l_2-2r \\
0 \le r \le \min\{m_1,l_2\}}} \lambda_{0, l}^{m_1, l_2} \alpha_l^{m_1,l_2}\alpha_l^{m_1,l_2*},\]  
with 
\begin{align*}\lambda_{0, l}^{m_1, l_2} &= \frac{[k+1]_q \text{Tet}_q\left[\begin{matrix} m_1 & l_1 & l \\ m_2& l_2&k
\end{matrix}\right]^2}{\theta_q(l_1,m_1,k)\theta_q(l_2,m_2,k) \theta_q(m_1,l_2, l)\theta_q(l_1,m_2,l)} \\
 &=\frac{[k+1]_q \left \{\begin{matrix} m_1 & l_1 & l \\ m_2& l_2&k
\end{matrix}\right\}_q^2\theta_q(l,l_1,m_2)\theta_q(l,m_1,l_2)}{\theta_q(l_1,m_1,k)\theta_q(l_2,m_2,k) [l+1]_q^2},
\end{align*} occurring with multiplicity $[l+1]_q$.
\end{rem}

\subsection{Remarks on the MOE additivity problem for certain $O_N^+$-TL-channels}
Given that we have, on the one hand, asymptotically sharp estimates on the MOE of the $O_N^+$-TL-channels $\Phi_k^{\bar l, m}$, $\Phi_k^{l, \bar m}$ (given by $H_{\min}(\Phi_k^{\bar l, m}),  H_{\min}(\Phi_k^{l, \bar m}) \sim \Big(\frac{l+m-k}{2}\Big)\log N$ - cf. Theorem \ref{thm-MOE-sharp}), and on the other hand,  we have exact formulae for the outputs $X_i =  \big( \Phi_{k_1}^{\bar l_1, m_1} \otimes \Phi_{k_2}^{l_2, \bar m_2}\big)(\rho_i^{k_1, k_2}) $ of entangled states under the tensor products of certain TL-channels, it is natural to ask whether one can obtain a strict inequality of the form 
\[
H(X_i) < \Big(\frac{l_1+m_1-k_1}{2}\Big)\log N + \Big(\frac{l_2+m_2-k_2}{2}\Big)\log N \qquad (\text{for suitable $i, k_j, l_j, m_j$}).
\]
If this were the case, we would have obtained deterministic examples of pairs of quantum channels which witness the non-additivity of their minimum output entropy.  

Unfortunately, however, extensive numerical evaluations of $H(X_i)$ for suitable parameter choices always yield inequalities of the form $H(X_i) - \Big(\frac{l_1+m_1-k_1}{2}\Big)\log N - \Big(\frac{l_2+m_2-k_2}{2}\Big)\log N >0$ with the difference going to zero as $N \to \infty$.  We see this as strong evidence that the pairs of quantum channels $\Phi_{k_1}^{\bar l_1, m_1}, \Phi_{k_2}^{l_2, \bar m_2} $ 
are not MOE strictly subadditive.

\section{Some Temperley-Lieb channels are not modified TRO-channels}\label{sec:TRO}

For a quantum channel $\Phi: B(H_A) \to B(H_B)$ with a Stinespring isometry $V : H_A \to H_B \otimes H_E$ the range space ${\rm Ran}V \subseteq H_B \otimes H_E$ is called a {\it Stinespring space} of $\Phi$. Note that the choice of isometry $V$ is not unique, but any associated Stinespring space is known to determine the channel $\Phi$. For this reason we will fix a Stinespring isometry $V$ and refer to the range ${\rm Ran}V$ as {\it the Stinespring space}. We say that the channel $\Phi$ is a {\it TRO-channel} if its Stinespring space is a {\it TRO}, i.e.  a{\it ternary ring of operators}. Recall that a TRO is a subspace $X$ of $B(H,K)$ for some Hilbert spaces $H,K$ such that $x,y,z \in X \Rightarrow xy^*z\in X$, i.e. closed under triple product. It is well-known that finite dimensional TRO's are direct sums of  rectangular matrix spaces with mutiplicity. Since the Stinespring space determines the channel it has been observed in \cite{GaJuLa16} that a TRO-channel $\Phi: B(H_A) \to B(H_B)$ is always of the following form: the channel $\Phi$ has a Stinespring space $X$ given by
	$$X = \oplus^M_{i=1}B(\Comp^{m_i}, \Comp^{n_i}) \otimes 1_{l_i} \subseteq B(H_E, H_B),$$
where
	$$H_E = \oplus^M_{i=1}\Comp^{m_i} \otimes \Comp^{l_i}\;\; \text{and} \;\; H_B = \oplus^M_{i=1}\Comp^{n_i} \otimes \Comp^{l_i}.$$
Moreover, we have $H_A = (X, \la \cdot, \cdot \ra_{H_A})$, where the inner product is given by $\la x, y \ra_{H_A} := \Tr_E(y^*x)$, $x,y \in X \subseteq B(H_E, H_B)$. Finally, the channel $\Phi$ is given by
	$$\Phi(|x\ra \la y|) = xy^*,\; x,y \in H_A=X \subseteq B(H_E, H_B).$$
Based on the above description we can define a variant of TRO-channels. We first fix a {\it symbol} $f\in B(H_E)$, i.e. a positive matrix with $\tau(f) := \frac{\Tr_E(f)}{d_E}=1$ and {\it strongly independent} of the right algebra $\mathcal{R}(X) = \text{span}\{x^*y: x,y\in X\}$. Here, we say that $x\in B(H_E)$ is {\it independent} of $\mathcal{R}(X)$ if $\tau(xy) = \tau(x)\tau(y)$ for all $y\in \mathcal{R}(X)$ and {\it strongly independent} of $\mathcal{R}(X)$ if $x^n$ is indepedent of $\mathcal{R}(X)$ for every $n\ge 1$. Then the {\it modified TRO-channel} $\Phi_f$ with the symbol $f$ is defined by
	$$\Phi_f : B(H_A) \to B(H_B),\;\; |x\ra \la y| \mapsto xfy^*.$$ The original TRO-channel $\Phi$ corresponds to the case of $\Phi_f$ with $f = 1_E$. It has been proved in \cite{GaJuLa16} that we have exact calculations for various capacities of $\Phi$ as follows.
	$$Q^{(1)}(\Phi) = P^{(1)}(\Phi) = Q(\Phi) = P(\Phi) = \log (\max_i n_i),\; \chi(\Phi) = C(\Phi) = \log (\sum_i n_i).$$
Moreover, we also have the following estimates for modified TRO-channels.
	$$Q^{(1)}(\Phi) \le Q^{(1)}(\Phi_f) \le Q^{(1)}(\Phi) + \tau(f \log f).$$
The same estimates hold for other capacities, i.e. we may replace $Q^{(1)}$ with $P^{(1)}, Q, P, \chi$ and $C$. Important examples of (modified) TRO-channels include random unitary channels using projective unitary representations of finite (quantum) groups and generalized dephasing channels \cite{GaJuLa16}.

In this section we prove that some TL-channels do not belong to the class of modified TRO-channels. Before we proceed to the details we need to be more precise about comparing two quantum channels.
\begin{defn}
Let $\Phi : B(H_A) \to B(H_B)$ and $\Psi: B(H_{A'}) \to B(H_{B'})$ be quantum channels with $d_B \le d_{B'}$. We say that $\Phi$ is equivalent to $\Psi$ if there is a unitary $U: H_A \to H_{A'}$ and an isometry $V: H_B \to H_{B'}$ such that
	$$V\Phi(U^* \rho\, U)V^* = \Psi(\rho),\;\; \rho \in B(H_A).$$
\end{defn}

We can find an example with minimal non-trivial dimensions. 

\begin{prop}\label{prop-non-TRO}
The $SU(2)$-TL-channel $\Phi^{\bar{2},1}_1$ is not equivalent to any modified TRO-channel.
\end{prop}
\begin{proof}
Since we have the associated isometry $\alpha^{2,1}_1:\Comp^2\rightarrow \Comp^3\otimes \Comp^2, \begin{array}{ll} |1\ra &\mapsto -\sqrt{\frac{2}{3}} |12\ra + \sqrt{\frac{1}{3}} |21\ra \\ |2\ra & \mapsto -\sqrt{\frac{1}{3}} |22\ra +\sqrt{\frac{2}{3}} |31\ra \end{array}$ (\cite{VK}),
we can see that channel $\Phi^{\bar{2},1}_1: B(\Comp^2) \to B(\Comp^2)$ maps $|1\ra \la 1| \mapsto \begin{bmatrix} 1/3 & 0 \\ 0 & 2/3 \end{bmatrix}$, $|1\ra \la 2|\mapsto \begin{bmatrix} 0 & - 1/3 \\ 0 & 0 \end{bmatrix}$, $|2\ra \la 1|\mapsto \begin{bmatrix} 0 & 0 \\ -1/3 & 0 \end{bmatrix}$ and $|2\ra \la 2|\mapsto \begin{bmatrix} 2/3 & 0 \\ 0 & 1/3 \end{bmatrix}$. Thus, we can observe that ${\rm Ran}\Phi^{\bar{2},1}_1 = B(\Comp^2)$, which is a full matrix algebra.

Let $\Phi_f$ be a modified TRO-channel with the parameters $n_i, m_i, l_i$, $1\le i\le M$ as above. Since we need to match the dimensions of the sender's Hilbert spaces we only have the following 3 possible cases. (1) $M=1$, $n_1 = 2$, $m_1=1$, (2) $M=1$, $n_1 = 1$, $m_1=2$ and (3) $M=2$, $n_1 = n_2 = m_1 = m_2 = 1$.

\vspace{0.3cm}

Case (1): The corresponding modified TRO-channel becomes (after identifying the orthonormal basis in a suitable way)
	$$\Phi_f : B(\Comp^2) \to B(\Comp^2) \otimes B(\Comp^{l_1}),\; |i\ra \la j| \mapsto |i\ra \la j| \otimes \frac{f}{l_1}.$$
If we assume that $\Phi^{\bar{2},1}_1$ is equivalent to $\Phi_f$, then there are a unitary $U: \Comp^2 \to \Comp^2$ and an isometry $V: \Comp^2 \to \Comp^2 \otimes \Comp^{l_1}$ such that
	$$V\Phi^{\bar{2},1}_1(U^* \rho\, U)V^* = \Phi_f(\rho),\;\; \rho \in B(H_A).$$
Since ${\rm Ran}\Phi^{\bar{2},1}_1 = B(\Comp^2)$ we also have ${\rm Ran}\Phi_f \cong B(\Comp^2)$ as a subalgebra of $B(\Comp^2) \otimes B(\Comp^{l_1})$, which forces $g := \frac{f}{l_1}$ to be a pure state. This implies that $g^2 = g$, so that $\Tr((|1\ra \la 1| \otimes g)^2) = \Tr(|1\ra \la 1| \otimes\frac{f}{l_1}) = 1$. However, the state $\rho' = \Phi^{\bar{2},1}_1(U^* |1\ra \la 1|U)$ can be easily shown to satisfy $\Tr((\rho')^2) = 5/9 \ne 1$. Since $X\mapsto VXV^*$ is a trace preserving map, we get a contradiction.

\vspace{0.3cm}

Case (2): The corresponding modified TRO-channel becomes 
	$$\Phi_f : B(\Comp^2) \to B(\Comp^{l_1}),\; |i\ra \la j| \mapsto \frac{f_{ij}}{l_1},$$
where $f = \begin{bmatrix} f_{11} & f_{12}\\ f_{21} & f_{22} \end{bmatrix} \in B(\Comp^2) \otimes B(\Comp^{l_1})$ with $f_{ij} \in B(\Comp^{l_1})$, $1\le i,j \le 2$. Since ${\rm Ran}\Phi^{\bar{2},1}_1 = B(\Comp^2)$ we know that $l_1 \ge 2$. We assume that there are a unitary $U: \Comp^2 \to \Comp^2$ and an isometry $V: \Comp^2 \to \Comp^2 \otimes \Comp^{l_1}$ such that $V\Phi^{\bar{2},1}_1(U^* \rho\, U)V^* = \Phi_f(\rho),\;\; \rho \in B(H_A)$ as before. In this case we have $\mathcal{R}(X) = B(\Comp^2)\otimes \Comp 1_{l_1}$. It is straightforward to check that independence of $f$ with respect to $\mathcal{R}(X)$ implies that $\Tr(f_{11}) = l_1$. We also know that $f^2$ is independent of $\mathcal{R}(X)$, which means that $\Tr((f^2)_{11}) = l_1$. However, we have
	$$l_1 = \Tr((f^2)_{11}) = \Tr(f^2_{11} + f_{12}f_{21}) \ge \Tr(f^2_{11}) = \frac{5}{9}l_1^2,$$
which is a contradiction. The above inequality is from $f^*_{12} = f_{21}$ and the last equality is from the fact that
	$$\Tr((\frac{f_{ij}}{l_1})^2) = \Tr((\rho')^2) = 5/9.$$
	
\vspace{0.3cm}

Case (3): The corresponding modified TRO-channel becomes 
	$$\Phi_f : B(\Comp^2) \to B(\Comp^{l_1+l_2}),\; \begin{bmatrix} a_{11} & a_{12}\\ a_{21} & a_{22} \end{bmatrix} \mapsto \left[\frac{f_{ij}}{\sqrt{l_i l_j}}\right]_{1\le i,j \le 2}, $$
where $f = \begin{bmatrix} f_{11} & f_{12}\\ f_{21} & f_{22} \end{bmatrix} \in B(\Comp^{l_1+l_2})$ with $f_{ij} \in B(\Comp^{l_j}, \Comp^{l_1})$, $1\le i,j \le 2$. Since ${\rm Ran}\Phi^{\bar{2},1}_1 = B(\Comp^2)$ we know that $l_1 \ge 2$. We assume that there are a unitary $U: \Comp^2 \to \Comp^2$ and an isometry $V: \Comp^2 \to \Comp^2 \otimes \Comp^{l_1}$ such that $V\Phi^{\bar{2},1}_1(U^* \rho\, U)V^* = \Phi_f(\rho),\;\; \rho \in B(H_A)$ as before. In this case we have $\mathcal{R}(X) = \Comp 1_{l_1} \oplus \Comp 1_{l_2} \subseteq B(\Comp^{l_1+l_2})$.  It is also straightforward to check that independence of $f$ with respect to $\mathcal{R}(X)$ implies that $\Tr(f_{11}) = l_1$. We also know that $f^2$ is independent of $\mathcal{R}(X)$, which means that $\Tr((f^2)_{11}) = l_1$. However, we have
	$$l_1 = \Tr((f^2)_{11}) = \Tr(f^2_{11} + f_{12}f_{21}) \ge \Tr(f^2_{11}) = \frac{5}{9}l_1^2,$$
where the last identity is from the fact that
	$$\Tr((\frac{f_{ij}}{l_1})^2) = \Tr(\begin{bmatrix}\frac{f_{ij}}{l_1} & 0 \\ 0 & 0\end{bmatrix}^2 ) = \Tr((\rho')^2) = 5/9.$$	
Thus, we can conclude that $l_1 = 1$, which actually means that $f_{11} = \Tr(f_{11}) = l_1 = 1$. Thus, we have $\Tr((\frac{f_{ij}}{l_1})^2) = 1 \ne 5/9$, so that we get a contradiction.

\begin{rem}
The canonical complementary channel $\tilde{\Phi}_f$ of a modified TRO-channel $\Phi_f$ can be written as follows.
	$$\tilde{\Phi}_f : B(H_A) \to B(H_E),\;\; |x\ra \la y| \mapsto \sqrt{f}y^*x\sqrt{f}.$$
Then, we can also show that the Temperley-Lieb channel $\Phi^{\bar{2},1}_1$ for $G = SU(2)$ is not equivalent to any canonical complementary channel $\tilde{\Phi}_f$ of a modified TRO-channel $\Phi_f$. This time the argument is easier since we only need to observe that ${\rm rank}(\tilde{\Phi}_f) \le 2$ in all the 3 possible cases in the proof of Proposition \ref{prop-non-TRO}.
\end{rem}

\end{proof}

\bibliographystyle{alpha}
\bibliography{TL-Capacity-MOE}

\end{document}